\documentclass[11pt]{article} 
\usepackage{amsmath,amssymb,amsthm,graphicx,algorithm}
\usepackage{algcompatible}
\usepackage{xspace}
\usepackage{nicefrac}
\usepackage{breqn}
\usepackage{authblk}

\title{Ornstein Isomorphism and Algorithmic Randomness}
\author[1]{Mrinalkanti Ghosh \thanks{mkghosh@ttic.edu}}
\author[2]{Satyadev Nandakumar \thanks{satyadev@cse.iitk.ac.in}}
\author[3]{Atanu Pal \thanks{palatanu@cse.iitk.ac.in}}
\affil[1]{Toyota Technological Institute at Chicago\\Chicago IL
  60637\\USA} 
\affil[2]{Department of Computer Science and Engineering\\
Indian Institute of Technology Kanpur\\Kanpur, UP, 208016, India} 
\affil[3]{Strand Genomics\\Bangalore\\India}

\newcommand{\ee}{\ensuremath{\mathrm e}}

\newcommand{\ve}{\ensuremath{\varepsilon}}

\newcommand{\rrr}{\ensuremath{\mathbb{R}}}

\renewcommand{\bar}[1]{\ensuremath{\overline{#1}}}

\let\nfrac=\nicefrac

\newcommand{\dotprod}[2]{\ensuremath{\left\langle #1, #2 \right\rangle}}
\newcommand{\ie}{i.e.,\xspace}

\makeatletter
\def\Pr#1{
  \@ifnextchar\bgroup%
   {\prwithdist{#1}}
   {\singlepr{#1}}
}
\def\singlepr#1{%
   \ensuremath{\mathbb{P}\left[ #1 \right]}
}
\def\prwithdist#1#2{%
   \@ifnextchar\bgroup
   {\superfancypr{#1}{#2}}
   {\ensuremath{\mathbb{P}_{#1}\left[ #2 \right]}}
}
\def\superfancypr#1#2#3#4{
   \ensuremath{\mathbb{P}_{#1}\left#3 #2 \right#4}
}
\makeatother

\makeatletter
\def\Ex#1{
  \@ifnextchar\bgroup%
   {\twoargsEx@aux{#1}}
   {\oneargEx@only{#1}}
}
\def\oneargEx@only#1{%
   \ensuremath{\mathbb{E}\left[ #1 \right]}
}
\def\twoargsEx@aux#1#2{%
   \@ifnextchar\bgroup
   {\superfancyex{#1}{#2}}
   {\ensuremath{\mathbb{E}_{#1}\left[ #2 \right]}}
}
\def\superfancyex#1#2#3#4{
   \ensuremath{\mathbb{E}_{#1}\left#3 #2 \right#4}
}
\makeatother

\makeatletter
\def\var#1{
  \@ifnextchar\bgroup%
   {\varwithdist{#1}}
   {\singlevar{#1}}
}
\def\singlevar#1{%
   \ensuremath{\mathrm{Var}\left[ #1 \right]}
}
\def\varwithdist#1#2{%
   \@ifnextchar\bgroup
   {\superfancyvar{#1}{#2}}
   {\ensuremath{\mathrm{Var}_{#1}\left[ #2 \right]}}
}
\def\superfancyvar#1#2#3#4{
   \ensuremath{\mathrm{Var}_{#1}\left#3 #2 \right#4}
}
\makeatother

\DeclareMathOperator{\bigomega}{\operatorname \Omega}
\DeclareMathOperator{\bigoh}{\operatorname O}

\newtheorem{theorem}{Theorem}
\newtheorem*{kstheorem}{The Kolmogorov-Sinai Theorem}
\newtheorem{maintheorem}[theorem]{(Main) Theorem}
\newtheorem{lemma}[theorem]{Lemma}
\newtheorem{corollary}[theorem]{Corollary}
\theoremstyle{definition}
\newtheorem{definition}[theorem]{Definition}
\newtheorem{example}[theorem]{Example}

\newcommand{\R}{\mathbb{R}}
\newcommand{\X}{(X, \mathcal{B}, \mu, T)}
\newcommand{\Y}{(Y, \mathcal{C}, \nu, S)}
\newcommand{\A}{\ensuremath{\mathcal{A}}\xspace}
\newcommand{\B}{\ensuremath{\mathcal{B}}\xspace}
\newcommand{\C}{\ensuremath{\mathcal{C}}\xspace}
\newcommand{\F}{\ensuremath{\mathcal{F}}\xspace}
\newcommand{\G}{\ensuremath{\mathcal{G}}\xspace}

\newcommand{\N}{\mathbb{N}}
\newcommand{\Z}{\mathbb{Z}}
\newcommand{\exampleend}{\hfill\ensuremath{\qed}}

\begin{document}
\maketitle

\begin{abstract}
In 1970, Donald Ornstein proved a landmark result in dynamical
systems, \emph{viz.}, two Bernoulli systems with the same entropy are
isomorphic except for a measure 0 set \cite{Ornstein70}. Keane and
Smorodinsky \cite{KeaneSmorodinsky79} gave a finitary proof of this
result. They also indicated how one can generalize the result to
mixing Markov Shifts in \cite{KeaneSmorodinsky79b}. We adapt the
construction given in \cite{KeaneSmorodinsky79} to show that if two
computable mixing Markov systems have the same entropy, then there is
a Schnorr layerwise lower semicomputable isomorphism defined on all
Schnorr random points in the system. Since the set of Schnorr random
points forms a larger set than the set of Martin-L\"of random points,
which is a measure 1 set, it implies the classical result for such
systems.

This result uses several recent developments in computable analysis
and algorithmic randomness. Following the work by Braverman
\cite{Braverman05}, Nandakumar \cite{Nan08}, and Hoyrup and Rojas
\cite{HoyRoj09} introduced discontinuous functions into the study of
algorithmic randomness. We utilize Hoyrup and Rojas' elegant notion of
layerwise computability and Miyabe's definition of Schnorr integrable
tests \cite{Miyabe13} to produce the test of randomness in our
result. 

We show that the result cannot be improved to include all points in
the systems - only trivial computable isomorphisms exist between
systems with the same entropy.
\end{abstract}

\section{Introduction}
In the Kolmogorov program for algorithmic randomness, Martin-L\"of
established that there is a smallest constructive measure 1 set, whose
objects are the set of individual random objects. Every effectively
computable probabilistic law, \emph{i.e.}  law which holds with
probability 1, specifies a ``majority rule''. Thus it is reasonable to
ask if every such law is satisfied by every individual random
object. This will \emph{a fortiori} imply the classical theorem, since
the set of random objects has probability 1. The effective versions
have more intuitive content, since they show that if any object fails
the particular law, then there is an algorithm which can ``bet'' and
win unbounded amounts of money on it.

Indeed, very general theorems like the Strong Law of Large Numbers
\cite{vanLambalgen}, the Law of Iterated Logarithm \cite{Vovk87}, and
Birkhoff's Ergodic Theorem \cite{Vyugin97}, \cite{Nan08},
\cite{GHR09b}, \cite{BDHMS12} have been effectivized. Prior to the
work of Braverman \cite{Braverman05}, only continuous functions were
considered. Following the work of Braverman, Nandakumar \cite{Nan08}
and Hoyrup and Rojas \cite{HoyRoj09a} have considerably broadened the
class of functions to deal with discontinuities, which has led to
considerably general theorems on the ergodic properties of random
objects in Bienvenu et al., and \cite{BDHMS12}, Franklin, Greenberg,
Miller and Ng \cite{FGMN12}. Recently, Hochman \cite{Hochman09} and
Hoyrup \cite{Hoyrup12} independently resolved the long-standing open
problem of the effectivization of the Shannon-McMillan-Breiman
theorem.

In a recent line of work, G\'{a}cs \cite{Gacs05}, and G\'{a}cs, Hoyrup
and Rojas \cite{GHR11}, \cite{HoyRoj09a}, \cite{HoyRoj09} have
extended the field of study of randomness to fairly general spaces
other than the finite alphabet spaces which have traditionally formed
the subject of algorithmic randomness. This also enables us to study
the \emph{relationships between} the random objects of different
probability spaces. In this paper, we utilize this theory to study
measure-preserving isomorphisms between effective dynamical
systems. We prove an effective version of the celebrated Ornstein
Isomorphism Theorem\cite{Ornstein70}, by adapting the finitary proof
of Keane and Smorodinsky \cite{KeaneSmorodinsky79}.

Consider two dynamical systems $\X$ and $\Y$\footnote{definitions in
  Section \ref{subsecn:kse}} where $X,Y$ are the sample space,
$\mathcal{B},\mathcal{C}$ the $\sigma$-algebras, $\mu$, and $\nu$ the
probabilities, and $T$ and $S$ the measure-preserving transformations
on $X$ and $Y$ respectively. A map $\phi: X \to Y$ is a factor map if
$\phi T(x) = S \phi (x)$ for almost every $x \in X$. If $\phi$ is
invertible then we say that $X$ and $Y$ are isomorphic. Isomorphisms
help us to categorize dynamical systems into classes of systems which
are essentially ``encodings'' of another system.

Kolmogorov and Sinai \cite{Kolmogorov:NITDS}, \cite{Sinai59}
introduced the notion of the \emph{entropy} of a dynamical system as
an invariant of an isomorphism. They showed that if two systems are
isomorphic to each other, then they have the same
\emph{Kolmogorov-Sinai} entropy. Ornstein and Weiss
\cite{OrnsteinWeiss07} show that this was a crucial insight -- in a
very broad sense, the Kolmogorov-Sinai entropy is the only invariant
of the isomorphism. The Kolmogorov-Sinai theorem brought a fresh
perspective to the study of dynamical systems. Formally, it justifies
viewing purely deterministic dynamical systems as having positive
entropy \cite{CFS} -- thus some deterministic systems can be viewed as
``random''.

The converse of the result, \emph{viz.} that systems with the same
Kolmogorov-Sinai entropy are isomorphic to each other, does not hold
in general (see Billingsley \cite{Billingsley65}). However, Ornstein
showed in a celebrated result, that if we restrict the systems to the
broad class of ``Bernoulli systems'', then equal entropy systems are
isomorphic to each other. Ornstein generalized this result to hold on
the class of ``finitely determined systems''. Numerous examples of
deterministic dynamical systems are isomorphic to the Bernoulli
system, which is intuitively the most random system possible. (For a
recent survey, see Ornstein \cite{Ornstein13}.)

However, the isomorphism Ornstein constructs is not continuous (it
cannot be continuous in general \cite{Petersen89}) and is not directly
amenable to the theory of algorithmic randomness. In 1979, Keane and
Smorodinsky gave a \emph{finitary} version of Ornstein isomorphism
theorem. A map is called \emph{finitary} if it is continuous except on
a measure 0 set. The concept involves viewing the underlying systems
as both probability and topological spaces. We adapt this proof to
establish our result.

Our main result of the paper is the following:
\begin{maintheorem}
If two effective mixing Markov systems have the same Kolmogorov-Sinai
entropy, then there is a ``layerwise computable'' isomorphism which is
defined on all Schnorr random objects of both the systems.
\end{maintheorem}

Hoyrup and Rojas \cite{HoyRoj09} have shown that layerwise computable
functions can be used to characterize Schnorr randomness. Hence
the above theorem will establish that there is an isomorphism which is
defined between the sets of Schnorr random objects in the two
systems. 


Further, in Section \ref{secn:computable_transformation}, we show that
this cannot be improved substantially -- if we insist on a computable
transformation which is defined on all points, then we have no
non-trivial isomorphism.

This work crucially employs the concept of layerwise computability,
which affords us the luxury of ignoring uncomputability of a function
on a large set of discontinuities. Our construction will diverge on
many non-random points. (For example, if a computable point $x$ has
only finitely many zeroes in its ``encoding'', then our map is
undefined at that point.)  This is an important difference from the
result of Keane and Smorodinsky (see Theorem 17 of \cite{KS:CFC}),
where the points of divergence of the construction are immaterial. We
show that for every Schnorr random object, the adapted
Keane-Smorodinsky construction converges -- in particular, in a
layerwise computable manner. Consequently there is a pointwise
isomorphism between the set of random objects in the two systems.

\section{Assumptions and Notations}
In this section we describe our notations for the proof developed in
section 4. In order to facilitate easy detection of parallel
constructs and differences between our proof and that of Keane and
Smorodinsky \cite{KeaneSmorodinsky79}, we closely follow notations of
the exposition in Chapter 6 of Petersen \cite{Petersen89}.

We are given two finite alphabet stationary mixing Markov systems
$\A~=~(\left(\Sigma_A\right)_{-\infty}^{\infty},P_A,T_A)$ and
$\B~=~(\left(\Sigma_B\right)_{-\infty}^{\infty},P_B,T_B)$ on alphabet
sets $\Sigma_A$ and $\Sigma_B$ respectively, with equal entropy.
Note that all the conditional probabilities are bounded away from $0$
or $1$.

Let $\varepsilon_r$ denote $\varepsilon_r=\frac{1}{2^r}$ for any
natural number $r$. We assume that the probabilities of the given
systems are computable. To be precise, we assume that we have a Turing
machine $M_A$ for the system \A (and $M_B$ for \B) so that given a
string $x\in\Sigma_A^*$ (correspondingly, $x\in\Sigma_B^*$) and a
natural number $n$, $M_A(x,n)$ ($M_B(x,n)$ for \B) returns a rational
number approximating the probability of a cylinder $x$ within
$\varepsilon_nP_A(x)$ of $P_A(x)$ ($\varepsilon_n\cdot P_B(x)$ for
\B). We denote this approximation by $P_A(x,n)$ and $P_B(x,n)$
respectively. Note that, since the dynamical systems are assumed to be
stationary, we do not care about the position of the
cylinder.\footnote{There is little difference between the requirements
  of having additive error of $\varepsilon_n$ and additive error of
  $\varepsilon_n\cdot P_A(x)$, except that the later is more
  convenient for our purpose.}

Given a probability vector $P$, we denote its entropy as $H(P)$. From
the above assumption, we can infer that the entropy of the systems is
computable, i.e., we have a Turing machine $M$, which on input $n$,
gives a ${\varepsilon_n}$ approximation of the entropy $H$.


\section{Overview of the construction}
First, we reduce the problem of construction of isomorphism between
two mixing Markov systems of equal entropy to one where two systems
have a common probability weight. We call this the Marker Lemma,
analogous to Keane and Smorodinsky. Our construction differs in that
all our systems are mixing Markov systems, unlike the Bernoulli
systems in \cite{KeaneSmorodinsky79}. This lemma allows us to assume,
without loss of generality, that the symbol 0 has identical
probability in the two systems.

A remark is due here about a false lead -- it may appear that if such
an intermediate construction succeeds, we can iterate the construction
and construct an isomorphism between the alphabets which \emph{a
  fortiori} yields a pointwise measure-preserving isomorphism. This is
not possible in general because the non-trivial cases of Ornstein
isomorphism are precisely when $|\Sigma_A|\ne|\Sigma_B|$, and we reach
an impasse when we have an odd number of symbols in one alphabet, and
an even number of symbols in the other.

Then, we construct an isomorphism between the random objects in two
mixing Markov systems \A and \C with equal entropy and with identical
probability for 0, in stages. First, for a random object $x$, we call
the pattern of 0s with all other symbols replaced by
$\text\textvisiblespace$ as the skeleton of $x$.  For $x \in \A$, we
identify potential images as those sequences $y \in \C$, which have
identical skeletons.  This is enabled by the effective
Skeleton Lemma. This is the first step to identify potential images of
$x$ under the isomorphism. We now restrict the choices available
progressively, until we remain with a unique image for $x$, through
the following stages.

Once we have identified sequences in \A and \C with identical
skeletons, we have to ``fill in'' the non-zero positions by producing
a measure-preserving bijection between equal length strings from the
two systems. The definition and technical results about these strings
form the ``effective filler lemma''. In this stage, we identify
``filled-in'' strings from \A and \C which could potentially be
isomorphically mapped to each other. The existence of strings in the
two systems with simultaneously the same length and approximately the
same entropy is a consequence of the asymptotic equipartition
property. This portion of our proof varies in an essential manner from
that of Keane and Smorodinsky.

This potential mapping between the strings of \A and \C can be
naturally modeled as a bipartite graph. Finally, we prove a version of
the Marriage Lemma to form the bijection between the strings in the
two sequences, which forms a basis for the construction of the
layerwise computable bijection between the two systems. In the limit,
we will map every random infinite sequence $x$ in the first system to
a unique random infinite sequence $y$ from the second and vice
versa. We will justify that the overall construction is a layerwise
computable function.

\subsection{Relevance of the Assumptions}
We crucially use the notion of Schnorr layerwise computable functions
from the theory of algorithmic randomness. Further, instead of the
effective Shannon-McMillan-Breiman theorem which holds for
Schnorr random points, we use the asymptotic equipartition
property of mixing Markov chains. We now broadly justify the
appropriateness of these assumptions.

Our algorithm relies on the fact that for any point in the support of
the isomorphism, we can find skeletons of any given rank. This is true
for all Schnorr random points, which is crucial in ensuring that our
construction is Schnorr layerwise computable. On the other hand, for
several computable points -- for instance, for periodic sequences,
skeletons of only finitely many lengths occur. Thus the set of points
where our algorithm diverges is \emph{dense}. Hence it seems difficult
to adapt topologically inspired notions of discontinuous functions
like that of Braverman \cite{Braverman05} or Nandakumar \cite{Nan08}
for our purpose, and measure-theoretic notions of computable
discontinuous functions like layerwise computability are considerably
more natural to deal with.

Second, the filler lemma for finding fillers for the skeleton relies
on the fact that \emph{for every Schnorr random point}, we can find
filler strings satisfying a certain entropy bound. The classical
Shannon-McMillan-Breiman theorem gives us only an almost everywhere
behavior which leaves the possibility that the construction may fail
for \emph{a nonempty measure 0 subset} of Schnorr random points. The
effective Shannon-McMillan-Breiman theorem of Hochman \cite{Hochman09}
and Hoyrup \cite{Hoyrup12} provides the assurance that we can find
such fillers for \emph{every} Martin-L\"of random point. However, we
need the stronger assurance that the fillers will exist for every
Schnorr random. In order to do this, we have to work directly with the
asymptotic equipartition property for mixing Markov chains. Even
though the classical property is a weaker version of the
Shannon-McMillan-Breiman theorem, here, this version gives us
sufficiently precise estimates for the Schnorr layerwise computable
function.


\section{Preliminaries}
In this section, we briefly explain the definition of concepts and
notation which we use in our result. First, we introduce the
background from dynamical systems, and second, that from algorithmic
randomness. 

\subsection{Kolmogorov-Sinai Entropy}
\label{subsecn:kse}
Kolmogorov \cite{Kolmogorov53} and Sinai \cite{Sinai59} introduced the
notion of the entropy of a \emph{transformation}, analogous to Shannon
entropy, which proved a fruitful tool in the classification of
dynamical systems. This notion is, in an essential sense, the only
invariant of a dynamical system -- all other natural invariants are
continuous functions of the entropy \cite{OrnsteinWeiss07}. We now
describe the notion of Kolmogorov-Sinai entropy.

A probability space is a triple $(X, \mathcal{B}, \mu)$, where $X$ is
a sample space, $\mathcal{B}$, a $\sigma$-algebra on $X$, and $\mu$, a
probability distribution on $\mathcal{B}$. Let $T: X \to X$ be a
measurable map. The transformation $T$ is called
\emph{measure-preserving} if for any measurable set $B \in
\mathcal{B}$, $\mu(T^{-1}B) = \mu(B)$. A measure-preserving map $T$ is
called an \emph{ergodic map} if every set $B \in \mathcal{B}$ where
$T^{-1}B = B$ has measure either 0 or 1.

\begin{definition}
A quadruple $(X, \mathcal{B}, \mu, T)$ where $(X, \mathcal{B}, \mu)$
is a probability space and $T: X \to X$ is an ergodic map, is called a
\emph{dynamical system}.
\end{definition}

We now proceed to the definition of entropy of a dynamical system. The
chief idea is to introduce a notion analogous to a finite alphabet.
Given any dynamical system $\X$, we can associate it with a process
involving finitely many states. Let $\alpha = (A_1, A_2, \dots, A_n)$
be a finite collection of measurable subsets of $X$ which are pairwise
disjoint except for measure 0 sets, and cover $X$ except possibly for
a measure 0 set. We can think of the partition containing $x \in X$ as
its $0^{\text{th}}$ ``character'' -- that is, if $x \in A_i$, then we
write $x[0] = i$.

The \emph{entropy of a partition $\alpha$} is defined to be $H(\alpha)
= - \sum_{i=1}^{n} \mu (A_i) \log_2 \mu(A_i)$. Then for any integer i,
$T^{-i} \alpha$ is the set $(T^{-i}(A_1), \dots, T^{-i}(A_n))$. This
set also partitions $X$, since $T$ is a measure-preserving
transformation. Now, we need to define concepts analogous to
``subsequences''. For this, we introduce the notion of refinement of
partitions.

If $\alpha = (A_1, \dots, A_n)$ and $\beta = (B_1, \dots, B_m)$ are
two partitions of $X$, then the \emph{join of the partitions}, $\alpha
\vee \beta$ is defined to be the partition
$$(A_i \cap B_j \mid i=1,\dots, n \;;\;j=1, \dots, m).$$

For any sequence of integers $i_1, \dots, i_k$, we then consider the
``least common refinement'' $\alpha[-k+1 \dots 0]$, denoted $\alpha
\vee T^{-1} \alpha \vee \dots \vee T^{-k+1}\alpha$. \footnote{The
  convention of starting from negative indices is standard in the
  literature on dynamical systems.} For any point $x \in X$, the cell
containing $x$ in this refinement represents the characters in
the positions $-k+1, \dots, -1, 0$.

Using this, for any $k \in \N$, we define the $k$-entropy of the
system as $H_k(\alpha) = \frac{1}{k} H(\alpha \vee T^{-1}\alpha \vee
\dots \vee T^{-k+1}\alpha)$, which represents the average entropy rate
of the letters $x[-k+1 \dots 0]$ of any point $x \in X$. Finally, the
asymptotic rate of entropy induced by the partition $\alpha$ is
defined $\lim_{k \to \infty} H_k(\alpha)$. This limit exists for every
stationary, in particular, ergodic systems.

\begin{definition}
The \emph{entropy} of the ergodic system $(X, \mathcal{B}, \mu, T)$
\emph{with respect to the partition $\alpha$} is $h(\alpha, T) =
\lim_{k \to \infty} \frac{1}{k} H_k(\alpha)$.

Let $\Pi(X)$ denote the set of all finite partitions of $X$. The
\emph{Kolmogorov-Sinai entropy of the transformation} $T$ is defined
to be
\begin{align}
\label{kse}
h(T) = \sup_{\alpha \in \Pi(X)} h(\alpha, T).
\end{align}
\end{definition}

The supremum in (\ref{kse}) is not easy to compute in
general. However, there is a case where the supremum is attained by a
fairly simple partition $\alpha$. We say that $\alpha$ is a
\emph{generator} of $\X$ if $\alpha \vee T^{-1} \alpha \vee \dots =
\mathcal{B}$ -- that is, if $\alpha$ generates the full
$\sigma$-algebra $\mathcal{B}$.  In this case, we have the famous
Kolmogorov-Sinai theorem.

\begin{kstheorem}{\cite{Kolmogorov53}, \cite{Sinai59}}
If $\alpha$ is a generator with respect to $T$, then $h(\alpha,T) =
h(T)$.
\end{kstheorem}

This has the consequence that for computable dynamical systems with a
computable generator, the entropy is computable. For a given dynamical
system, from now on, we will assume that a generating partition is
given and thus we can view the dynamical system as an alphabet process
with left shift being the ergodic transform from the space to itself.

The notion of entropy was then used to settle an open question. This
involves the relationship between two dynamical systems $\X$ and $\Y$. 

\begin{definition}
Two dynamical systems $\X$ and $\Y$ are said to be \emph{isomorphic}
to each other if there is a measure preserving invertible map 
$\phi:X \to Y$ such that $\phi T(x) = S \phi(x)$ for $\mu$-almost
every $x \in X$.
\end{definition}

Now let us observe the following:
$\phi(x)[i]=(S^i(\phi(x)))[0]=(\phi(T^i x))[0]$.  Hence as long as we
can compute the central coordinates of the images for $T^ix$ $($ for
all $i\in\Z)$, we can compute the isomorphism $\phi(x)$. So, from now
on we only wish to determine the central alphabet of the image under
the isomorphism.

Kolmogorov proved the following theorem.

\begin{theorem}(Kolmogorov \cite{Kolmogorov53})
If two dynamical systems are isomorphic, then they have the same
Kolmogorov-Sinai entropy.
\end{theorem}

He used this to negate the existence of a specific isomorphism by
showing that the systems involved had different
entropies. \cite{Billingsley65}

The converse of the question does not hold in general. To see some
examples, see Section 5 of Billingsley \cite{Billingsley65}. However,
Ornstein showed a powerful result: that for a large class of systems,
called \emph{finitely determined systems}, the converse of Kolmogorov's
theorem is true -- that is, if two such systems have the same entropy,
then there is an isomorphism between them \cite{Ornstein70b}. This
construction cannot be ``continuous'' in general. In a more specific
context, Keane and Smorodinsky \cite{KeaneSmorodinsky79} gave a
finitary construction between two Bernoulli systems of the same
entropy. We introduce the terminology below.

\begin{definition}
An isomorphism is called \emph{finitary} if for almost every $x \in X$
there exists a $j \in \N$ such that for every $x' \in X$, such that
$x[-j\dots 0 \dots j] = x'[-j\dots 0 \dots j]$, we have that $(\phi
x)[0] = (\phi x')[0]$.
\end{definition}

Note that this $j$ exists only for a measure 1 subset of $X$, and not
necessarily for every point in it. Also, the $j$ depends on the
specific $x$ that we choose. Keane and Smorodinsky proved that for
Bernoulli systems, Ornstein's construction can be made finitary.

\begin{theorem}\cite{KeaneSmorodinsky79}
If $\X$ and $\Y$ are two Bernoulli systems with the same
Kolmogorov-Sinai entropy, then there is a finitary isomorphism between
$\X$ and $\Y$.
\end{theorem}

In our work, we show that the above construction can be utilized to
construct a layerwise lower semicomputable isomorphism between the
sets of algorithmically random objects of two computable mixing Markov
dynamical systems. To introduce this strengthening, we now give an
overview of the setting of algorithmic randomness.

\subsection{Algorithmic Randomness and Layerwise Tests}
\label{sub-sec:algorithmic_randomness}

One of the important applications of the theory of computing is in the
definition of \emph{individual} random objects, finite strings and
infinite binary sequences in a mathematically robust way -- first
defined using constructive measure theory by Martin-L\"{o}f
\cite{MartinLof}. In this paper, we mention a recent generalization of
the theory of algorithmic randomness to fairly general spaces, namely,
computable metric spaces. G\'{a}cs \cite{Gacs05}, and G\'{a}cs, Hoyrup
and Rojas, in a series of works \cite{GHR11}, \cite{HoyRoj09a} have
shown that there are universal tests of randomness in these general
spaces. In this paper, we will deal with the Cantor space, where most
of the general theory is not directly required.  However, we need this
theory for two specific purposes -- first, we need the definition of a
computable probability space.
Second, the general theory of computable
metric spaces is used to define the notion of layerwise computability
\cite{HoyRoj09a}, \cite{HoyRoj09} which provides a more flexible way
to determine whether an element of the space is algorithmically
random. This theory plays a crucial role in our result.

\begin{definition}
  A space $(X, d)$ is called a \emph{computable metric space} if it
  satisfies the following.
  \begin{enumerate}
    \item $X$ is separable -- \emph{i.e.}, it has a countable dense
      subset $\mathcal{S}$.
    \item $\mathcal{S} = \{ s_i \mid i \in \N\}$ is a computably
      enumerable set.
    \item For any $s_i, s_j \in \mathcal{S}$, $d(s_i, s_j)$ are
      uniformly computable real numbers.
  \end{enumerate}
\end{definition}

If $x \in X$ and $r > 0$, then the metric ball $B(x,r)$ is the subset
of $X$ of points at less than $r$ distance from $x$. We consider a set
of ideal balls $\mathcal{N} = \{B(s,q) \mid s \in \mathcal{S}, q \in
\mathcal{Q}\}$. The set of ideal balls is associated with a canonical
computably enumerable numbering $\mathcal{N} = \{B_i \mid i \in \N\}$.

\begin{example}
The unit interval $[0,1]$ endowed with the Euclidean metric, is a
computable metric space. The set of dyadic rationals $\{ \frac{m}{2^k}
\mid m, k \in \N\}$is a computably enumerable dense subset
$\mathcal{S}$.  The set of canonical balls is then uniquely
determined.

Pick any computable enumeration of the rationals.  Then it is routine
to utilize this to produce a canonical enumeration of the set of ideal
balls.  \exampleend
\end{example}

\begin{definition}
An \emph{effectively open set} is an open set $U$ such that there is a
computably enumerable set of indices $E \subseteq \N$ with $\cup_{j
  \in E} B_j = U$.
\end{definition}

Thus effectively open sets are the analogues of computably enumerable
sets. Similarly, we can define notions of computability on these
metric spaces. A function $f : X \to [-\infty, \infty]$ is \emph{lower
  semicomputable} if the sets $f^{-1}(q,\infty]$ are uniformly
  effectively open. A function $f: X \to [-\infty, \infty]$ is
  \emph{upper semicomputable} if $-f$ is lower semicomputable, and is
  computable if it is both upper and lower semicomputable.

\begin{definition}
Let $(X, d, \mathcal{S})$ be a computable metric space. A Borel
probability measure $\mu$ on $X$ is \emph{computable} if the
probability of any finite union of canonical balls is computable.
\end{definition}

In other words, there is a machine, which for every $\epsilon$ and
every finite union of cylinders $C$, returns a rational number with
$\epsilon$ of the probability of $C$.\footnote{ This is a more
  restricted notion than that considered in Hoyrup and Rojas
  \cite{GHR09b}.} 

\begin{example}
For the previous example, the Borel measure generated by specifying
that $\mu((x,y]) = |y-x|$ is a computable probability measure.
\exampleend
\end{example}

Hoyrup and Rojas \cite{HoyRoj09a} prove an effective Prokhorov theorem
for computable probability measures on computable metric spaces, which
is the basis for their new definition of algorithmic randomness. For
this, first we need the notion of a layerwise lower semicomputable
function.

A \emph{Martin-L\"of test} $O$ is a sequence of uniformly effectively
open sets $O_n$ such that for every $n \in \N$, $P(O_n) <
\frac{1}{2^n}$. A point $x$ is said to be Martin-L\"of random if for
every Martin-L\"of test $O$, $x \notin O_n$ for some $n$. If $P$ is a
computable probability measure, then the set of Martin-L\"of points
has $P$ measure 1.

Every computable probability space $(X,P)$ also has a universal
Martin-L\"of test -- that is, there is a Martin-L\"of test $U$ such
that $x \in X$ is Martin-L\"of random if and only if there is an $n
\in \N$, $x \notin U_n$.

\begin{definition}\cite{HoyRoj09a}, \cite{HoyRoj09b}
Let $(X,P)$ be a computable probability space. Let $U$ be a universal
Martin-L\"of test for $P$. Then the sequence of compact sets $\langle
K_n \rangle_{n=0}^\infty$ where $K_n = X - U_n$ for every $n \in \N$,
is defined as the \emph{layering} of the space. For every $n \in \N$
is called the \emph{$n^{\text{th}}$ layer} of the space.
\end{definition}

\begin{definition}
A lowersemicomputable function $f:X \to \R$ is called \emph{layerwise
  lowersemicomputable} if it is uniformly computable on $\langle K
\rangle_{n=1}^{\infty}$.
\end{definition}

The layerwise lower semicomputable functions may be undefined on every
point that is not Schnorr random. This is important since our
construction diverges on many (but not necessarily all) nonrandom
points.

\begin{definition}
A \emph{layerwise integrable test} is a layerwise lower semicomputable
function $t: X \to [0, \infty]$ such that $\int t d\mu$ is finite.

A point $x \in X$ is Martin-L\"of random if for every layerwise
integrable test $t$, we have $t(x) < \infty$.
\end{definition}

The integrable function can be thought of as a martingale
process. Thus a point is Martin-L\"of random if no layerwise
lowersemicomputable martingale can win unbounded money on
it. We deal with a slightly stronger notion, \emph{viz.}, Schnorr
layerwise computability. We use a definition due to Miyabe
\cite{Miyabe13}. 

\begin{definition}{(Miyabe \cite{Miyabe13})}
\label{schnorr_it}
A Martin-L\"of integrable test $f$ is a \emph{Schnorr integrable test}
if there is a computable sequence of rational-valued step functions
$\langle s_n \rangle$ converging to $f$ pointwise such that $||s_{n+1}
- s_{n}||_1 < 2^{-n}$.
\end{definition}


We will construct an isomorphism between two spaces which is layerwise
lower semicomputable. Then we argue that the composition of the
layerwise test on the domain and the isomorphism constitutes a
layerwise test on the range.



\section{Construction of the Isomorphism}
\subsection{Effective marker lemma - intermediate Markov system}

In this section, given two systems \A and \B with same entropy,
we designate one alphabet from each of \A and \B, say $0$ and $1$.
We then construct a mixing Markov chain $\C=(\Sigma_C,P_C)$ with
designated alphabets $0,1$ and with the following properties:
\begin{enumerate}
 \item $P_A(0)=P_C(0)$ and $P_A(00)=P_C(00)$, \ie probabilities
 of cylinders containing only $0$s are same in \A and \C.
 \item Similarly for the system \B and alphabet $1$: $P_B(1)=P_C(1)$
 and $P_B(11)=P_C(11)$.
 \item Entropy of \C is same as that of \A and \B.
 \item $P_C(\omega)$ is computable for any $\omega\in\Sigma_C^*$.
\end{enumerate}

Here the conditions $3$ and $4$ are somewhat opposing in nature: Since
the entropy of \A and \B can be arbitrary large we may want to set the
probabilities of \C somewhere close to uniform distribution (while
maintaining probabilities of $0$s and $1$s). But due to computable
nature of probabilities of \A and \B we have only approximates
available for the target entropy (the entropy is also computable
). Because the gradient of the entropy function near the uniform
distribution is almost horizontal, we may need to make substantial (
multiplicative) change in probabilities of system \C to match the
target entropy within acceptable error. But this breaks the
computability requirement of the probabilities of \C.

However, we are able to manage the two competing requirements
simultaneously.  We give a recursive procedure to get approximate
probabilities for the system \C. First we make sure that the
probabilities of $0$s and $1$s are matched to that of \A and \B within
acceptable error. Then we enforce a lower and upper bound on the
conditional probabilities of the system \C. The lower-bound enforces
that the system \C is (fast enough) mixing -- we require this lower
bound in further sections. The upper-bound is carefully chosen so that
the entropy of \C can match that of \A and \B while the gradient is
steep enough so that we only need to make small change in
probabilities to make the required change in entropy. This allows us
to produce a sequence of approximates to the probability distribution
of \C while maintaining all the above mentioned requirements.

The formal details of the outline mentioned above is as follows:
Let $H$ be the entropy of the systems $\A$ and $\B$ and the memory of
the Markov processes $\A$ and $\B$ be 1. Let $0$ be a symbol in $\A$
which minimizes the following conditional probability: $P_A(x[1]=a\mid
x[0]=a)$, where $a$ is in $\A$ ( breaking the ties arbitrarily from an
approximation of probabilities up to a small enough error ).  Similarly,
let $1$ be a symbol in $\B$ which minimizes the following conditional
probability: $P_B(x[1]=b\mid x[0]=b)$, where $b$ is in $\B$.

We construct the intermediate system $\C$ to be of memory
$1$. We let the alphabet of the system $\C$ to be $\Sigma_C=
\{0,1,\dots,c\}$, where $c$ is determined later.

We consider the set of $\Pi_0$ probability distributions on
$\Sigma_C^{2}$ such that in each distribution, each element of
$\Sigma_C^{2}$ has probability $>0$ (in fact we will ultimately use a
stronger lower bound). For brevity, let us denote $P_{xa}$ to be the
probability of the string $xa$, where $x,a\in\Sigma_C$, in the
distribution $P\in \Pi_0$. Let us also denote $P_x=\sum_{a\in\Sigma_C}
P_{xa}$. Note that when $P$ is a distribution which describes a Markov
process, the entropy of the process is defined as (the conditional
entropy conditioned on first step): $h(P) =\sum_{xa\in\Sigma_C^{2}}
P_{xa} \log \left(\frac{P_x}{P_{xa}}\right)$\footnote{ For simplicity
let us assume that the base of the logarithm is $\ee$ -- this only
changes entropy by a constant factor. One can perform similar computation
by appropriately multiplying the constant $\log_2\ee$. }.

Let $p_0=P_A(\omega[1]=0\mid\omega[0]=0)$ and
$p_1=P_B(\omega[1]=1\mid\omega[0]=1)$.  Let $\alpha$ be the value
which, if assigned to $P_C(\omega[1]=0\mid \omega[0]=x)$ for all
$x\in\Sigma_C\setminus\{0\}$, yields $P_C(0)=P_A(0)$. A simple
computation shows: $\alpha=\frac{P_A(0)(1-p_0)}{1-P_A(0)}$. Similarly,
let $\beta$ be the value of that needs to be assigned to
$P_C(\omega[1]=1\mid \omega[0]=x)$ for all
$x\in\Sigma_C\setminus\{1\}$ to ensure $P_C(1)=P_B(1)$. Also similar
equality for $\beta$ holds: $\beta=\frac{P_B(1)(1-p_1)}{1-P_B(1)}$

Let $\gamma=\alpha+\beta$. Let $\eta,\delta>0$ be two parameters to be
determined later.  In what follows, we restrict ourselves to
probability distributions in $\Pi\subset \Pi_0$ which have the
following properties:
\begin{itemize}
 \item $P_{x0}=P_x\cdot \alpha$ for all $x\in\Sigma_C\setminus \{0\}$.
 \item $P_{00}=P_A(\omega[0]=0,\omega[1]=0)$.
 \item $P_{x1}=P_x\cdot \beta$ for all $x\in\Sigma_C\setminus \{1\}$.
 \item $P_{11}=P_B(\omega[0]=1,\omega[1]=1)$.
 \item For all $x\in\Sigma_C$ and $a\in\Sigma_C\setminus \{0,1,c\}$,
   $P_{xa}\le \delta P_{x}$, \ie $\delta$ is an upper-bound on the
   conditional probabilities on all but the symbols $0,1$ and $c$.
 \item For all $x\in\Sigma_C$ and $a\in\Sigma_C\setminus \{0,1,c\}$,
   $P_{xa}\ge \eta P_{x}$, \ie $\eta$ is a lower-bound on the
   conditional probabilities on all but the symbols $0,1$ and $c$.
\end{itemize}
We observe that $\Pi$ is closed under convex combinations, \ie $\Pi$ is
convex. We exhibit a distribution in $\Pi$ to show that it is
non-empty. For any distribution in $\Pi$ the probabilities of
cylinders containing only $0$s matches that of $\A$ and similarly
probabilities of cylinders containing only $1$s matches that of
$\B$. Also note that only fixing the conditional probabilities is
enough to specify the distribution, since the conditional
probabilities specify an unique stationary distribution. For the
construction of $\C$, we only restrict ourselves to distributions in
$\Pi$. We call a distribution $\mu\in \Pi$ to be an interior
distribution if all the inequalities are satisfied strictly.

We let $\eta=\frac{1}{c(c-3)}$.
Consider the distribution $Q$ corresponding to the
following conditional probabilities:
\[
 \frac{Q_{xa}}{Q_x}=
 \begin{cases}
  \eta & \mbox{if } a\ne 0,1,c\\
  \alpha & \mbox{if } x\ne 0\mbox{ and } a=0\\
  \beta & \mbox{if } x\ne 1\mbox{ and } a=1\\
   p_0 & \mbox{if } x=0 \mbox{ and } a=0\\
   p_1 & \mbox{if } x=1 \mbox{ and } a=1\\
  1-\sum_{b\ne c}\frac{Q_{xb}}{Q_x} & \mbox{if } a=c\\
 \end{cases}
\]
Note that the stationary distribution for this is given by 
$$Q_a=
\begin{cases}
 P_A(0) & \mbox{if } a=0\\
 P_B(1) & \mbox{if } a=1\\
 \eta & \mbox{if } a\ne0,1,c\\
 1-P_A(0)-P_B(1)-(c-3)\eta & \mbox{if } a=c\\
\end{cases}.$$
So, by construction $Q\in \Pi$ (and hence $\Pi$ is non-empty). Now
the entropy of $Q$ is: 
\begin{align*}
  h(Q) = & \sum_{xa} Q_{xa} \log\frac{Q_x}{Q_{xa}}\\
   = & \sum_{x\ne 0,1} \left( Q_x \alpha\log\frac1\alpha + Q_x
  \beta\log\frac1\beta + \sum_{a\ne 0,1,c} Q_{xa} \log\frac1\eta
  + Q_{xc}\log\frac{1}{1-\gamma-(c-3)\eta}\right) \\
  & + \left(
  Q_0p_0\log\frac{1}{p_0}+ Q_0
  \beta\log\frac1\beta + \sum_{a\ne 0,1,c} Q_{0a} \log\frac1\eta
  + Q_{0c}\log\frac{1}{1-p_0-\beta-(c-3)\eta}
  \right)
  \\
  & + \left(
  Q_1p_1\log\frac{1}{p_1}+ Q_1
  \alpha\log\frac1\alpha + \sum_{a\ne 0,1,c} Q_{1a} \log\frac1\eta
  + Q_{1c}\log\frac{1}{1-p_1-\alpha-(c-3)\eta}
  \right)
  \\
  = & (1-Q_0)\alpha\log\frac1\alpha + Q_0p_0\log\frac{1}{p_0}
  + (1-Q_1)\beta\log\frac1\beta + Q_1p_1\log\frac{1}{p_1}\\
  & + (c-3)\eta\log\frac1\eta + (1-Q_0-Q_1) (1-\gamma-(c-3)\eta)
  \log\frac{1}{1-\gamma-(c-3)\eta}\\
  & + Q_{0c}\log\frac{1}{1-p_0-\beta-(c-3)\eta}
  + Q_{1c}\log\frac{1}{1-p_1-\alpha-(c-3)\eta} \\
  \overset{c\to\infty}{\longrightarrow} & \left((1-Q_0)\alpha\log\frac1\alpha + Q_0p_0\log\frac{1}{p_0}
    + (1-Q_1)\beta\log\frac1\beta + Q_1p_1\log\frac{1}{p_1}\right.\\
  & \left.+ (1-Q_0-Q_1) (1-\gamma)\log\frac{1}{1-\gamma} + Q_{0c}
  \log\frac{1}{1-p_0-\beta} + Q_{1c}\log\frac{1}{1-p_1
  -\alpha}\right)\\
\end{align*} 
for our choice of $\eta$. Now we notice that by considering $\A^m$ and $\B^m$ ( \ie
new alphabets are $m$-tuple of old alphabets and new shifts to be old shift repeated
$m$ time ), the limit value decreases ( as $p_0, p_1$ doesn't increase and further,
$Q_0, Q_1$ and hence $\alpha,\beta$ decreases ) while the entropies of the systems
$\A^m$, $\B^m$ increases ( becomes $mH$ ). An isomorphism between $\A^m$ and $\B^m$
yields an isomorphism between $\A$ and $\B$. We choose a suitable value of $m$ so that
limit value of $h(Q)$ becomes strictly less than the entropy of $\A^m$ and $\B^m$. We
construct isomorphism between $\A^m$ and $\B^m$ and for simplicity of notations we 
ignore $m$ from here on.

Let $\delta=\frac{1}{M(c-3)}$ for some large enough $M$ to be determined later.
Consider the following distribution $R$, as given by the conditional probabilities:
\[
 \frac{R_{xa}}{R_x} =
 \begin{cases}
  \delta & \mbox{if } a\ne 0,1,c\\
  \alpha & \mbox{if } x\ne 0\mbox{ and } a=0\\
  \beta & \mbox{if } x\ne 1\mbox{ and } a=1\\
   p_0 & \mbox{if } x=0 \mbox{ and } a=0\\
   p_1 & \mbox{if } x=1 \mbox{ and } a=1\\
  1-\sum_{b\ne c}\frac{Q_{xb}}{Q_x} & \mbox{if } a=c\\
 \end{cases}
\]
As in previous case, the stationary distribution is given by: 
$$R_a=
\begin{cases}
 P_A(0) & \mbox{if } a=0\\
 P_B(1) & \mbox{if } a=1\\
 \delta & \mbox{if } a\ne0,1,c\\
 1-P_A(0)-P_B(1)-(c-3)\delta & \mbox{if } a=c\\
\end{cases}.$$
Following the previous computation, we see that
\begin{align*}
 h(R) = & (1-R_0)\alpha\log\frac1\alpha + R_0p_0\log\frac{1}{p_0}
  + (1-R_1)\beta\log\frac1\beta + R_1p_1\log\frac{1}{p_1}\\
  & + (c-3)\delta\log\frac1\delta + (1-R_0-R_1) (1-\gamma-(c-3)\delta)
  \log\frac{1}{1-\gamma-(c-3)\delta}\\
  & + R_{0c}\log\frac{1}{1-p_0-\beta-(c-3)\delta}
  + R_{1c}\log\frac{1}{1-p_1-\alpha-(c-3)\delta} \\
  = & \Theta\left(\frac{\log ({(c-3)M})}{M}\right).
\end{align*}
Hence for a fixed $M$, we can choose a large enough $c$ so that $h(R)$ strictly
surpasses $H$.

We treat $h$ as a function $h:\xspace \rrr^{c^{2}} \to \rrr$, where we index
the co-ordinates with corresponding string in $\Sigma_C^{2}$.
Now:
\begin{align*}
\left(\nabla h\right)_{xa} & =\frac{\partial h}{\partial P_{xa}}\\
& = \frac{\partial}{\partial P_{xa}} \left( \sum_{\substack{b\in\Sigma_C\\b\ne a}}
      P_{xb} \log\frac{P_x}{P_{xb}} \right) + \frac{\partial}{\partial P_{xa}}
	\left(P_{xa}\log\frac{P_x}{P_{xa}}\right)\\
& = \sum_{\substack{b\in\Sigma_C\\b\ne a}} P_{xb}\cdot\frac{P_{xb}}{P_x}
      \cdot\frac{1}{P_{xb}} + \log\frac{P_x}{P_{xa}} + P_{xa}\cdot \frac{P_{xa}
	}{P_x}\cdot \frac{P_{xa}-P_x}{P_{xa}^2}\\
& = \frac{P_x-P_{xa}}{P_x} + \log\frac{P_x}{P_{xa}} + \frac{P_{xa}-P_x}{P_{xa}}\\
& = \log\frac{P_x}{P_{xa}}.
\end{align*}

Consider a third distribution (not necessarily in $\Pi$) defined as:
\[
 \frac{U_{xa}}{U_x} =
 \begin{cases}
  \alpha & \mbox{if } x\ne 0\mbox{ and } a=0\\
  \beta & \mbox{if } x\ne 1\mbox{ and } a=1\\
   p_0 & \mbox{if } x=0 \mbox{ and } a=0\\
   p_1 & \mbox{if } x=1 \mbox{ and } a=1\\
  \frac{1-\frac{U_{x0}}{U_x}-\frac{U_{x0}}{U_x}}{c-2} & \mbox{if } a\ne 0,1\\
 \end{cases}
\]
As earlier, we note that $U_{a}=\frac{1-P_A(0)-P_B(1)}{c-2}$ for any
$a\ne 0,1$.

For a given interior distribution $P$ in $\Pi$ and let $P'$ defined as:
$P'=(1-\ve)P+\ve U$ for small enough $\ve$ so that $P'\in \Pi$. ( Since
the upper-bound on the conditional probabilities are strictly
satisfied, adding very small quantity to it doesn't violate the
inequalities ). Now:
\begin{align*}
 \dotprod{\nabla h(P)}{P'-P} = & \ve \sum_{xa} \left(U_{xa}\log\frac{P_x}{P_{xa}} - 
 P_{xa}\log\frac{P_x}{P_{xa}} \right)\\
 = & \ve \sum_{xa} U_{xa}\log\frac{P_x}{P_{xa}} - \ve h(P) \\
 = & \ve\left(\sum_{x\ne 0}U_{x0}\log{1}{\alpha} + \sum_{\substack{xa\\a\ne 0,1,c}} U_{xa}
 \log\frac{P_x}{P_{xa}} + U_{00}\log\frac{1}{p_0} + U_{11}\log\frac{1}{p_1}\right.\\
 & \left.+ \sum_{x} U_{xc} \log\frac{P_x}{P_{xc}}\right) - \ve h(P)\\
 \ge & \ve \sum_{\substack{xa\\a\ne 0,1,c}} U_{xa} \log\frac{1}{\delta} -\ve h(P)\\
 = & \ve\left((c-3)\cdot\frac{1-P_A(0)-P_B(1)}{c-2}\log{M(c-3)} - h(P)\right)
\end{align*}
Note that in $\Pi$, $R$ has highest entropy (since it has most
balanced probability distribution and $R$ is majorized by every other
distribution in $\Pi$). As we saw earlier, $h(R)=\Theta\left(\frac{\log
  c+\log M}{M}\right)$. So for a suitable choice of $c,M$ (large
enough) $\dotprod{\nabla h(P)}{P'-P}=\bigomega(\ve)$.

For the same $P$, let $P''=(1+\ve)P-\ve U$. For small enough $\ve$, we
also note that $P''\in \Pi$. By an argument similar to the previous one,
we get: $\dotprod{\nabla h(P)}{P-P''}=\bigomega(\ve)$.

Also note that both $P',P''$ changes probability values by at most
$\bigoh(\ve)$.

Hence we can choose a $c$ and corresponding $M$. We can choose a
starting distribution in $\Pi$ such that entropy of the distribution is
close (up to, say $\ve$) to $H$. Then for all large enough $n$, we can
get $n$-th approximate for probabilities, the limit of which, defines
the probability distribution of the intermediate system $\C$. By choice
of $\Pi$, we note that probabilities of cylinders containing only $0$s
or only $1$s are as desired. By appropriately modifying the
probability distributions at each step (\ie choosing $P'$ or $P''$ for
appropriate $\ve$), the entropy of $\C$ can also be made to be equal to
$H$.
\subsection{An Effective Skeleton Lemma}{\label{sub-sec:skeleton}}
We can now consider two systems \A and \C with $P_A(0) = P_C(0)$. We
consider those pairs $(x,y) \in \A \times \C$ such that their patterns
of zeroes are ``similar'', and progressively restricting this set, we
will finally ensure an isomorphism for every pair of random
sequences. For this, we now introduce the notion of a
\emph{skeleton}. The skeleton of a finite string is the string we
obtain by mapping any non-zero symbol in it to a special character,
say $\text{\textvisiblespace}$. Consequently, if the patterns of 0s in
two finite strings $x \in \Sigma_A^*$ and $y \in \Sigma_C^*$ are
identical, then their skeletons are identical. In this subsection, we
prove an effective version of Keane and Smorodinsky's Skeleton Lemma
\cite{KeaneSmorodinsky79} (see also Chapter 6, Lemma 5.3 in
Petersen\cite{Petersen89}). What goes in the blank spaces is called a
filler.

The strategy that we adopt in the isomorphism is to map sequences $x
\in \Sigma_A^\infty$ to sequences $y \in \Sigma_C^\infty$ with
identical skeletons. The first stage in the construction is to
identify the set of potential pairs of infinite sequences with
identical skeletons. To this end, we now define the notion of a
skeleton of rank $r$, $r \in \N$, and show that Schnorr random
sequences in any system have skeletons of all ranks. Owing to the fact
that we have only approximation of probabilities of mixing Markov
systems, we consider a different setting for skeletons and later,
their fillers, from the one considered in \cite{KeaneSmorodinsky79}
for Bernoulli process.

Assume that we have a sequence of positive integers $N_0 < N_1 <
\dots$. (This sequence will be fixed when we discuss the filler lemma,
where we establish that it can be computed layerwise.) For a skeleton
of rank $r$ centered at position $i$ in a sequence $x$, we look for
the shortest substring centered at $x[i]$ starting and ending with
$N_r$ (or more) consecutive $0$s. We replace all non-zero symbols with
blanks. We replace all non-zero symbols with blanks. We further replace the
maximal blocks of $0$s\footnote{ Here we deviate from the
  original construction.} of length $1$ ( \ie stand alone $0$s ) with blanks.
\begin{definition}
Let $x \in A^\Z$. A \emph{skeleton $S_{x,r,i}$ of rank $r$ in $x =
  [.....x_{-2} x_{1} x_{0} x_{1} x_{2}.....]$} is defined as
follows. Starting from $x[i]$, pick the shortest string of the form
$0^{n_0} \text{\textvisiblespace}^{\ell_1}
0^{n_1}\ \dots\ \text{\textvisiblespace}^{\ell_k}0^{n_k}$ such that
the following hold.
\begin{itemize}
\item Each $\ell_i$ is at least $1$, $(1 \le i \le k)$.
\item Each $n_i$ is at least $2$, $(1 \le i \le k)$.
\item $n_i < N_r$ for all $1 \le i \le k-1$. Further, both $n_0$ and
  $n_{k}$ are greater than or equal to $N_r$.
\end{itemize}
\end{definition}
Thus, except for the extremities of the skeleton of rank $r$,
there is no contiguous block of $0$s longer than $N_r$. Also, it is
routine to see that a rank-$r$ skeleton can be uniquely decomposed
into skeletons of rank $r-1$ \cite{Petersen89}.

We now show that the skeleton of every Schnorr random object in has
skeletons of every rank $r$ (with respect to any predetermined
sequence $N_1 < N_2 < \dots$ of numbers) while having sufficiently
many blanks in between. This is an effective version of the Skeleton
Lemma in \cite{KeaneSmorodinsky79}.

\begin{definition}
The \emph{length of the skeleton} $S_{x,r,i}$, denoted
$\ell(S_{x,r,i})$, is defined as follows.
$$\ell(S_{x,r,i}) = \left|\lbrace i \mid x_{i} \neq 0,\quad i \in
S_{x,r,i}\rbrace\right|$$
\end{definition}

\begin{lemma}[Schnorr Layerwise Skeleton Lemma]
\label{lemma:skeleton}
Let $\langle L_r \rangle_{r=1}^{\infty}$ be a computable increasing
sequence of positive integers. Then there is a Schnorr layering
$\langle K'_r \rangle_{r=1}^\infty$ of $\A$ and an increasing sequence
of positive integers $\langle N_r \rangle_{r=0}^\infty$ uniformly
computably enumerable in $\langle K'_r \rangle_{r=1}^\infty$ such that
for every $r \in \N$ and every $x \in K'_r$, the following hold.
\begin{itemize}
\item There is a skeleton centered at $x[0]$ delimited by $N_r$ many
  zeroes.
\item The central skeleton centered at $x[0]$ and delimited by $N_r$
  many zeroes, has length at least $L_r$.
\end{itemize}
\end{lemma}

\begin{proof}
Define $K'_r = \{ x \in X \mid \ell(S_{x,r,0}) \ge L_r \}$. (Note that
in this step, we choose $N_1, \dots, N_r$ to determine the rank-$r$
skeleton.)Thus $K'_r$ contains all points $x$ such that their
``central skeleton'' of rank $r$ contains at least $L_r$ many spaces.

Consider
\[K' = \cup_{n=1}^{\infty} \cap_{r=n}^{\infty} K'_r,\]
the set of points in $X$ such that for large enough ranks $r$, a
skeleton of rank $r$ contains at least $L_r$ many \textvisiblespace\xspace symbols. We
form a Schnorr integrable test which attains infinity on each element
in $K^{'c}$.

Any $x$ in $K^{'c}$ has either of two properties -- first, $x$ does
not have any skeleton of rank $r$ (or above), and second, for every
$n$, there exists some rank $r \ge n$ such that $x$ has a central
skeleton having less than $L_r$ many spaces. We will form a Schnorr
layerwise integrable functions which will attain $\infty$ on $x$ in
either of these cases.

{\bf Case I.} Suppose $x$ has no central skeleton of rank $r$ or
more. By the pigeonhole principle, there is some rank $r'<r$ such that
a rank $r'$ skeleton appears infinitely often in $x$. Suppose $r'$ is
the highest rank which appears infinitely often in any skeleton of
$x$, including non-central skeletons.

Let the left zero extremity (analogously, the right zero extremity) of
a string $w$ be the longest block of zeroes at the left end
(correspondingly the right end) of $w$. (These may, of course be
empty.) Let $ZE: \Sigma_A^* \to \{0\}^*$ be the function which returns
the shorter among the left zero extremity and right zero extremity.

Consider the following function defined on cylinders of
$\Sigma_A^\infty$. The function $f: \Sigma_A^* \to [0,\infty)$ is
  defined by
\begin{align*}
f(\lambda) &= 1\\
f(a_1\ w\ a_2) &= \begin{cases}
                 \frac{1}{(1-P_A(0w0\mid w))} f(w) &\text{ if } |ZE(w)|=N_{r'}
                 \text{ and } a_1a_2 \ne 00\\
                 0                       &\text{ if } |ZE(w)|=N_{r'} 
                 \text{ and }	  a_1a_2 = 00\\
                 f(w)                    &\text{ otherwise}
                 \end{cases}
\end{align*}

Define the function $S: A^\Z \to [0, \infty)$ by $S(x) = \sup_{n}
  f(x[-n \dots 0 \dots n])$. Since $P_A$ is computable, we can
  conclude that $S$ is layerwise lower semicomputable.

For infinitely many $n$, a skeleton of rank $r'$ will appear as the
extremities of $x[-n\dots 0 \dots n]$. Hence the subsequent bits on
the left and the right cannot both be 0. In this case, $f(x[-n-1\dots
  0 \dots n+1]) > f(x[-n \dots 0 \dots n])$. Thus, $S(x) = \infty$.

We observe that $\int f(\lambda) dP_A = 1$. Similarly, on any cylinder
$w$, if $w$ does not have extremities of the form $0^{N_r}$, then
$f(a_1 w a_2) = f(w)$, and we have
$$\sum_{a_1 a_2 \in \Sigma_A^2} f(a_1 w a_2) P_A(a_1wa_2 \mid w) =
f(w) \sum_{a_1a_2 \in \Sigma_A^2} P_A(a_1a_2 \mid w),$$ which is
$f(w)$. If $w$ ends in extremities of the form $0^{N_r}$, then
$$\sum_{a_1 w a_2 \in \Sigma^2\setminus \{00\}} f(a_1 w a_2) P_A(a_1 w
a_2 \mid w) = f(w) \frac{[1-P_A(0w0\mid w)]}{1-P_A(0w0\mid w)},$$
which is $f(w)$ as well. So we have that
$$f(w) P_A(w) = \sum_{a_1a_2 \in \Sigma^2} f(a_1\ w\ a_2)
P_A(a_1\ w\ a_2).$$

Thus, it follows that
$$\int S(x) dP_A = \int \limsup_n f(x[-n \dots 0 \dots n]) dP_A
\le\sup_n \int f(x[-n \dots 0 \dots n]) dP_A= 1,$$ where the
inequality follows by Fatou's lemma.

To show that the layering above is a Schnorr layering, we show that
$S$ is $L^1$-computable. We construct a computable sequence $\langle
s_n \rangle_{n \in \N}$ of computable step functions pointwise
converging to $S$ where for all $n$, $|| s_{n+1} - s_{n} ||_1 \le
\theta^{n}$. \footnote{Without loss of generality, the $2^n$ in the
  Definition \ref{schnorr_it} may be replaced by any computable
  inverse exponentially decaying bound.} The step function $s_n:
\Sigma^* \to \Sigma^*$ is defined by
\begin{align*}
s_n(axb) = \begin{cases}
\max_{0 \le i \le |x|-1} f(x[-i \dots i]) &\text{if }|x| \le n\\
s_n(x)      &\text{otherwise}
\end{cases}
\end{align*}
It is clear that $s_n \to S$ pointwise. Now, $s_n(\omega)$ and
$s_{n+1}(\omega)$ differ only on those points where $\omega[-n \dots
  n])$ has $0^{N_r}$ at both ends. Let us designate the set of strings
$x \in \Sigma^n$ which end with $0^{N_r}$ as $G$.
\begin{align*}
\int |s_{n}(x) - s_{n+1}(x)| dP 
&= \sum_{x \in G, a,b \in \Sigma} |s_{n+1}(axb) - s_{n}(axb)| 
P(axb|x) P(x)\\
&= \frac{s_n(x)}{1-P(0x0|x)} (1-P(0x0|x)) P(x)+ 
0 \times P(x)\\
&\le s_n(x) \theta^{n}.
\end{align*}
It follows that $S$ is a Schnorr layerwise computable function.

The above argument shows that the set of sequences which lack a
particular rank can be captured by a Schnorr layerwise integrable
function. Now we show that sequences which lack some rank can be
similarly captured by a Schnorr layerwise integrable test, by taking a
convex combination of the individual tests, even though in general,
there is no universal Schnorr test. Denoting the test for a particular
rank by $S_r$, consider the test $S = \sum_{r=1}^{\infty} 2^{-r}
S_r$. If there is an $r$ and an $\omega \in A^\Z$ such that
$S_r(\omega) = \infty$, then $S(\omega) = \infty$ as well. Since each
$S_r$ is monotone non-decreasing in the length of the string, so is
$S$. Also, $\int S dP_A = \int \sum_{r=1}^\infty 2^{-r} S_r$, which is
finite. We now show that $S$ is $L^1$ computable.

For $n \in \N$, consider the rational step function $s_n =
\sum_{i=1}^n 2^{-i} s_{i,n}$. As $n \to \infty$, this converges to $S$
pointwise, since each individual sequence $\langle s_{i,n} \rangle_{n
  \in \N}$ converges to $S_i$ pointwise. We now have to show that for
all $n \in \N$, $||s_{n+1} - s_{n} ||_1$ has a computable upper bound
decaying exponentially in $n$, uniformly over $n$.

Now, since each $\langle s_{i,n} \rangle_{n \in \N}$ is monotone
non-decreasing in $n$, it follows that $|s_{n+1}(x) - s_{n}(x)| =
s_{n+1}(x) - s_{n}(x)$. By using the estimates on the individual
$s_i$s, we get the following bound. For every $x \in A^{n+1}$, we have 
\begin{align*}
|s_{n+1}(x) - s_{n}(x)| &= 
\left|\sum_{i=1}^{n+1} 2^{-i} s_{i,n+1}(x) -
 \sum_{i=1}^{n} 2^{-i} s_{i,n}(x)\right|\\
&= \sum_{i=1}^{n} \left[2^{-i} (s_{i,n+1}(x) - s_{i,n}(x))\right] + 
   2^{-(n+1)} s_{i,n+1}(x).
\end{align*}
Hence,
\begin{align*}
\sum_{x \in \Sigma_A^{n+1}} |s_{n+1}(x) - s_{n}(x)| P_A(x) &\le
\sum_{i=1}^n 2^{-i} \theta^{n+1} 
+ 2^{-(n+1)} s_{i,n+1}(x) \theta^{n+1}\\
&< 
2 \theta^{n+1} 
+ 2^{-(n+1)} \frac{1}{\theta^{n+1}} \theta^{n+1}\\
&= 2\theta^{n+1}+2^{-(n+1)}.
\end{align*}

{\bf Case II.} Now suppose that for every $n$, there is a central
skeleton in $x$ of rank $r \ge n$ such that $\ell(S_{x,r,0}) <
L_r$. This implies that within at most $L_r (N_r - 1)$ characters
around $x_0$, the block $0^{N_r}$ will occur in $x$. 

Consider the function $g_r: A^* \to [0,\infty)$ defined by
\begin{align*}
g^k_r(\lambda) &= \frac{1}{2^r L_r(N_r-1)}\\
g^k_r(a_1\ w\ a_2) &= \begin{cases}
  \frac{1}{P_A(0w0 \mid w)} g_r(w) 
  &\text{ if } k \le |w| < L_r(N_r-1) \text{ and } a_1 a_2 =00\\ 
  0 &\text{ if } k \le |w| < L_r(N_r-1) \text{ and } a_1 a_2 \not= 00\\
  g(w) &\text{otherwise.}
  \end{cases}
\end{align*}

As in case I, we can verify that for all cylinders $w$,
$$g^k_r(w)P_A(w) = \sum_{a_1a_2 \in \Sigma^2} g^k_r(a_1\ w\ a_2)
P_A(a_1\ w\ a_2).$$

Consider the function $g_r: A^* \to [0,\infty)$ defined by
$$g_r = \sum_{k=1}^{L_r(N_r-1)} g^k_r.$$

We know that if $x$ has a deficient rank $r$ at length $k$, then 
$$g_r(x) \ge \frac{1}{2^r} \frac{1}{L_r (N_r-1) P_A(0)^{N_r}} \ge 1$$
if we choose large $N_r$ in a suitable manner.

Finally, consider the aggregate function $S: \Sigma_A^\infty \to
[0,\infty)$ defined by $S = \sum_{r=1}^{\infty} \sup_n g^k_r (x[-n
    \dots n])$. Then, as in case I, we see that $S$ is Schnorr
  layerwise lower semicomputable and integrable. Since by assumption
  $x$ has infinitely many $r$ for which $g_r$ attains at least $1$, we
  have that $S(x) = \infty$.
\end{proof}

We will now proceed to choose this sequence of $L_r$s that is assumed
in Lemma \ref{lemma:skeleton}.



\subsection[Effectively determining L_r and Filler
  lemma]{Effectively determining $L_r$ and Filler
  lemma}{\label{sub-sec:filler}} In the last subsection, we assume
that we have a sequence $L_0 < L_1 < \dots$ of natural numbers. For
every $i, r \in \N$ and $x \in \A \cup \C$, a skeleton in $x$ of rank
$r$ at position $i$ was the shortest string centered at $x[i]$ and
delimited by the earliest appearance of at least $N_r$ many zeroes and
at least $L_r$ many spaces.  We now see how to determine this sequence
in a Schnorr layerwise lower semicomputable manner.

We define a sequence of $\langle L_r\rangle_{r=1}^\infty$ for the
lengths of the skeletons of rank $r$ inductively. We choose the
sequence $\langle N_r \rangle_{r=1}^{\infty}$ such that a skeleton of
rank $r$ has length at least $L_r$. We compute the lengths $L_r$
layerwise, in such a way that properties analogous to the asymptotic
equipartition property hold for the skeletons of rank $r$ for every
Schnorr random sequence. This will allow us to construct a provably
isomorphic map between \A and \C.

Let $\eta_r=\min_{D\in\{A,C\}}\min_{\substack{a\in\Sigma_D,
    b\in\Sigma_D}} P_D(x[1]=a \mid x[0]=b,r)$ and $\theta_r$ be the
corresponding maximum. For a mixing Markov chain, these will be
bounded away from 0 and 1.  Here, $\eta_r$ and $\theta_r$ are computable.

We pick a strictly increasing sequence $\langle
L_r\rangle_{r=1}^\infty$ such that:\footnote{ Here we deviate from the
   original construction.}  \[\lim_{r\to\infty}
\frac{1}{\eta_r}2^{-L_r(\varepsilon_{r-1}-\varepsilon_r)}=0\].

Let $\F(S) \subseteq \Sigma_A^\ell$ denote the set of fillers for $S$
in \A.  Let $Z_S$ denote the indices of $0$s in $S$ and let the blanks
be in positions $B = (s_1,s_2,\dots s_\ell)$.  Given a filler
$F\in\F(S)$ and an index set $I \subseteq B$, let $\langle I,F,S
\rangle$ denote the cylinder generated by setting $0$s from $S$ and
setting $i^{th}$ position for $i \in I$ with the corresponding symbol
in the filler $F$.

For an $n \in \N,n\ge r$,\footnote{We define $J(F,n)$ only when $n\ge
  r$} we define an equivalence relation $\sim_n$ for error bound
$\varepsilon_n$ on $\F(S)$ and denote equivalence class of $F$ by
$\tilde{F}_n$. We decide a subset of places $J(F,n) \subseteq \{s_1,
s_2,\dots, s_\ell\}$ for each $F$ and declare $F\sim_n F'$ if
$J(F,n)=J(F',n)$ and $F$ agrees with $F'$ on $J(F,n)$.

For a fixed $n\ge r$ and $F$, we define $J$ inductively on the rank of
the skeleton. For a skeleton $S$ of rank 1 and length $\ell$, we
proceed as follows. For a $k \le l$, let $B_k$ denote $(s_1, \dots,
s_k)$. Pick the largest positive integer $k$, $k \le \ell$ such that
$P_A(\ \langle B_k ,F,S\rangle,\ n\ )$ is at least $\nfrac{3}{2\eta_1}
2^{-(\ell+|Z_S|)(H-\epsilon_1)}$. Then, let $J(F,n) = Z_S \cup B_k$.

Now, for a rank $r\ge 2$ skeleton $S$ and $F\in\F(S)$, we do the
following: Let us assume that $S=S_1\times S_2\times \dots \times S_t$
is the skeleton decomposition of $S$ where each $S_i$ is of rank
$r-1$. Also let $F_1,F_2,\dots,F_t$ are the corresponding fillers
which coincides with $F$. We assume that we have determined $J(F_i,
n\log 3t)$ inductively for each $F_i$. Let $J_0(F,n)=\cup_{i=1}^t
J(F_i, n\log 3t)$.\footnote{The purpose of $n\log 3t$ will be clear in
  lemma \ref{lemma:prod-of-society}} These are the positions in $S$
which have already been determined in the previous rank.

Also, let $\{s_1,\dots s_\ell\}\setminus J_0(F,n) =(t_1\dots
t_u)$. These are the positions in the skeleton $S$ which have not been
fixed by any rank $r-1$ sub-skeletons. Let $T_k = (t_1, \dots, t_k)$,
for $k \le \ell$. Then, we set $J(F,n)=Z_S\cup J_0(F,n)\cup T_k$,
where $k \le u$ is the largest index such that $P_A(\ \langle T_k \cup
J_0(F,n), F,S \rangle, n\ )$ exceeds
$\frac{(1+\varepsilon_r)}{\eta_r}2^{-(\ell+|Z_S|)(H_{}-\varepsilon_r)}$.
Here, $\eta_r/(1+\varepsilon_r)$ is a pessimistic approximation of
true minimum conditional probability of an alphabet.

Let $x=0^{l_1}x_10^{\ell_2}x_20^{\ell_3}\dots 0^{\ell_{t}}
x_{t}0^{\ell_{t+1}}$ be a string where $\ell_i>m$ for all $1\le i\le
t+1$.  Let
$$P'_A(x,n)=\frac{\prod_{i=1}^t
    P_A(0^{\ell_{i}}x_i0^{\ell_{i+1}},n\log 3t)}{\prod_{i=2}^t
    P_A(0^{\ell_i})}.$$
By the Markov property, 
\begin{equation}
\label{eqn:probmult}
  \left|P_A(x)-P'_A(x,n)\right|\le \varepsilon_n P'_A(x,n).
\end{equation}
So, $P'_A(x,n)$ can be used in place of $P_A(x,n)$ but for the fact we
cannot compute $P_A(0^{\ell_i})$ exactly. But we use the essentially
multiplicative nature of $P'_A$ and that it approximates $P_A$ in the
proof of Lemma \ref{lemma:prod-of-society}.  The approximation is as
follows: $|P_A(x,n)-P'_A(x,n)| \le 2\varepsilon_r P_A(x,n)$ -- this is
the essential observation which makes our construction possible.

Also, we note that for a given $F\in\F(S)$ and an integer $n$,
$J(F,n)\subseteq J(F,n+1)$. In other words, if we decrease the error
bound in estimation of probability the equivalence relation can only
get finer. Similar relations holds for \C. Then the asymptotic
equipartition property of mixing Markov chains yields the following
bounds.

\begin{lemma}[Filler Lemma]
  \label{lemma:filler}
There is a Schnorr layering $\langle K''_p\rangle_{p=1}^\infty$ such
that for every $n$, there is a large enough $r\ge n$ such that for
every skeleton $S$ of rank $r$ and length $\ell$ corresponding to
$x\in K''_r$, we have the following.
\begin{enumerate}
\item For all $F\in\F(S)$, $P_A\left(\tilde{F}_r,r\right)\ge
  (1+\varepsilon_r)2^{-L(H-\varepsilon_r)}$
\item For all $F\in\F(S)$ except maybe on a set of measure
  $\varepsilon_n$:
  \begin{enumerate}
  \item
    $P_A(\tilde{F}_r,r) <\frac{1+\varepsilon_n}{\eta_n}
    2^{-L(H_{}-\varepsilon_n)}$
  \item $\frac{1}{L}|J(F,r)| > 1-\frac{2}{|\log_2
    \theta_r|}\varepsilon_n$
  \end{enumerate}
\end{enumerate}
where $L=\ell+|Z_S|$.
\end{lemma}
\begin{proof}
From the asymptotic equipartition property for Markov chains (see, for
example, Chapter 1 of Khinchin \cite{Khinchin57}), we know that there
is a Schnorr layering $\langle K''_p\rangle_{p=1}^\infty$ defined
below.

For all $p$ there is a $k_p$ so that for all $k\ge k_p$,
$\Sigma_{A}^k=K''_p\cup (K''_p)^c$ is the largest set with the
following properties:
\begin{itemize}
\item $P_A(K''_p)\ge 1-\varepsilon_{p}$
\item For each $x\in K''_p$ we have 
$$\frac{1-\varepsilon_p}{\eta_p}2^{-k(H_{}+\varepsilon_{p})}< P_A(x,p)<
    \frac{1+\varepsilon_p}{\eta_p}2^{-k(H_{}-\varepsilon_{p})}.$$
\end{itemize}
Since the last condition can be decided by examining $x[-p \dots p]$
and $P_A$ is computable, it follows that $P_A(K''_p)$ is computable,
uniformly in $p$.

Now given an $n$, let $n'$ be such that $2\varepsilon_{n'}\le
\varepsilon_n$.  Let $r\ge n+1$ be such that $L_r\ge k_{n'}$.  Such an
$r$ exists, since $\{L_r\}$ is an increasing sequence. For brevity,
we denote $J_0(F) \cup \{t_1, \dots, t_w\}$ by $J_1$.
\begin{enumerate}
\item Let $J(F,r)=Z_S\cup J_1$. Then 
  \begin{eqnarray*}
  P_A(\tilde{F}_r,r) & = & P_A(\langle
  J_1,F,S\rangle,r)\\  & = &
  P_A(J_1,F,S\rangle,r) \times
  P_A(F[t_w]|\langle
  J_1,F,S\rangle,r)\\ & \ge &
  \frac{1+\varepsilon_r}{\eta_r}2^{-L(H_{}-\varepsilon_{r})}
  \times P_A(F[t_w]|\langle
  J_1,F,S\rangle,r)\\ & \ge &
  (1+\varepsilon_r)2^{-L(H_{}-\varepsilon_{r})}.
  \end{eqnarray*}
  where the inequality before the last follows from the definition of
  $J(F,r)$.
\item
  \begin{enumerate}
  \item If $|J(F,r)|<L$, then by the definition of $J(F,r)$, we
    have $$P_A(\tilde{F}_r,r)<\frac{1+\varepsilon_r}{\eta_r}
    2^{-L(H_{}-\varepsilon_r)}.$$
    If $|J(F,r)|=L$, then
    $\tilde{F}_r=F$. But $|F|=L\ge L_n\ge k_{n'}$ and hence $$P_A(F,r)
    <\frac{1+\varepsilon_{n'}}{\eta_{n'}} 2^{-L(H_{}-\varepsilon_{n'})}$$
    unless $F\in (K''_r)^c$ and
    $P_A((K''_r)^c)\le\varepsilon_{n'}<\varepsilon_n$.
    Since $\varepsilon_r<\varepsilon_n'<\varepsilon_n$, we have
    $\frac{1+\varepsilon_r}{\eta_r}
    2^{-L(H_{}-\varepsilon_r)}<\frac{1+\varepsilon_{n'}}{\eta_{n'}}
    2^{-L(H_{}-\varepsilon_{n'})}<\frac{1+\varepsilon_{n}}{\eta_{n}}
    2^{-L(H_{}-\varepsilon_{n})}$ (from definition of $L$).
  \item Without loss of generality, assume that (a) holds.  (Otherwise
    we already have that such $F$ has to be in $\varepsilon_r$ measure
    set.) Let $L-|J(F,r)|\ge 2L\varepsilon_n/|\log_2 \theta_r|$. Then,
    \begin{eqnarray*}
      P_A(F,r)
      &=&  P_A(\tilde{F}_r,r) \times
      \prod_{i\not\in J(F,r)} P_A(F[i]|\langle 
      J_1,F,S\rangle,r)\\  
      &\le& P_A(\tilde{F}_r,r) \cdot \theta_r^{2L\varepsilon_n/|\log_2
        \theta_r|}\\ 
      &<&
      \frac{1+\varepsilon_{n'}}{\eta_{n'}}2^{-L(H_{}-\varepsilon_{n'})}
      2^{-2L\varepsilon_n}\\   
      &<&
      \frac{1+\varepsilon_{n}}{\eta_{n}}2^{-L(H_{}-\varepsilon_{n})}
      2^{-2L\varepsilon_n}\\   
    \end{eqnarray*}
    We use the inequality $\theta^{1/|\log_2 \theta|}\le 2^{-1}$. 
    In this case $F$ must belong to the set $(K'_r)^c$ 
    of measure less than $\varepsilon_{n'}$.  Hence the set on which $L$
    can violate the bound has measure 
    $<2\varepsilon_{n'}\le \varepsilon_n$.
  \end{enumerate}
\end{enumerate}
\end{proof}



\subsection{Societies and Marriage Lemma}{\label{sub-sec:marraige}}
Once we have determined the filler alphabets and filler probabilities
for \A and \C, we are now in a position to start building the
isomorphism between cylinders from \A and \C which have identical
skeletons. Each cylinder in \A has multiple possible matches in \C and
conversely. We model this as a bipartite graph with the filled-in
skeletons from \A forming the left set of vertices, and those from \C
forming the right set. The presence of an edge represents a potential
match between the corresponding vertices. We obtain this by a minor
variant of Keane and Smorodinsky's marriage lemma, where the variation
is forced by the fact that we have only an approximation of the actual
probabilities of the vertices.

Let us assume we are given two probability space $(\Omega_1,\mu_1)$,
$(\Omega_2,\mu_2)$, with both $\Omega_1$ and $\Omega_2$ finite. A
\emph{society} or a \emph{knowledge relationship} is a map
$f:\Omega_1\to2^{\Omega_2}$ so that for all $X\subseteq\Omega_1$, we
have $\mu_1(X) \le \mu_2 ( f(X) )$ where $f(X)$ is defined in the
natural way. When the underlying probabilities are clear from context,
we denote a society as $f:\Omega_1\leadsto \Omega_2$. Now consider the
undirected knowledge graph constructed out of the knowledge
relationship, with vertices set $\Omega_1\cup\Omega_2$ and edge set
$E=E_1\cup E_1^{-1}$ where $E_1=\{(a,b)\in\Omega_1\times\Omega_2:b\in
f(a)\}$ .  Note that the knowledge graph is bipartite by
definition. Now we define a couple of notions which provides us with
the tools necessary for defining isomorphism:

\begin{definition}[Join of societies]
Given societies $f_i:\Omega_{i,1}\leadsto \Omega_{i,2}$ for
$1\le i\le j$, we define their join $f:\Omega_{1,1}\times
\Omega_{2,1}\times\ldots\Omega_{j,1}
\overset{prod}\longrightarrow\Omega_{1,2}\times   
\Omega_{2,2}\times\ldots\Omega_{j,2}$ as a map $f:\Omega_{1,1}\times
\Omega_{2,1}\times\ldots\Omega_{j,1}\to2^{\Omega_{1,2}\times
\Omega_{2,2}\times\ldots\Omega_{j,2}}$ where $(\omega_1,\omega_2,\ldots \omega_j)
\in f(\nu_1,\nu_2,\ldots \nu_j)$ for $\omega_i\in\Omega_{i,2}$,
$\nu_i\in\Omega_{i,1}$ 
iff $\omega_i\in f_i(\nu_i)$.
\end{definition}

\begin{definition}[$\varepsilon$-robust]
Consider a society $f$ between probability spaces $(\Omega_1,\mu_1)$,
$(\Omega_2,\mu_2)$.  Consider the undirected knowledge graph
$G=V_1\cup V_2 \cup \cdots \cup V_w$ where $V_i$s are connected
components of $G$. Given an $\varepsilon>0$, society $f$ is called
$\varepsilon$-robust if for all $1\le i\le m$, for all $X\subset
V_i\cap \Omega_1$ and for all $Y\subset V_i \cap \Omega_2$, we have:
\begin{align*}
\mu_1(X) (1+\varepsilon) &\le \ \mu_2 (f(X)) ( 1-\varepsilon )\\
\mu_2(Y) (1+\varepsilon) &\le \ \mu_1(f^{-1}(Y)) (1-\varepsilon).
\end{align*}
\end{definition}

It is easy to see that for $\varepsilon>0$, an $\varepsilon$-robust
society is a society.

Note that we only consider proper subsets $X$ and $Y$ in the above
definition, since
$\mu_1(V_i\cap\Omega_1)=\mu_2(V_i\cap\Omega_2)$. This easily follows
from the fact that $f$ and $f^{-1}$ are societies. Also note that a
society $f$ is $\varepsilon$-robust iff the dual of the society
$f^{-1}$ is $\varepsilon$-robust.

A society is \emph{minimal} if the removal of any edge will violate
the condition for a society. In the construction of an isomorphism, we
consider various minimal sub-societies of given societies. Now since
we only have some approximation of probabilities of vertices, we have
to be careful while removing edges from knowledge graph to construct
minimal sub-society. The next lemma shows that it is enough to
consider $\varepsilon$-robust minimal societies for our purpose.

\begin{lemma}
  \label{lemma:robust-suffice}
 Given a society $f$ between probability spaces $(\Omega_1,\mu_1)$,
 $(\Omega_2,\mu_2)$ and a minimal sub-society $g$, there is an
 $\varepsilon > 0$ so that $g$ is $\varepsilon$-robust.
\end{lemma}
\begin{proof}
We know that the minimal sub-society $g$ is generated by a
joining\footnote{In the literature, the joining operation is also
  known as coupling.}, say $\mu$ - that is, a joint distribution $\mu$
on $\Omega_1 \times \Omega_2$ such that $\mu_1$ and $\mu_2$ are its
marginals(see Chapter 6 of \cite{Petersen89}). Consider the knowledge
graph $G$ for the society $g$. Note that $G$ is a finite graph. Let
$G=V_1 \cup V_2 \cup \cdots \cup V_w$, where $V_i$s are the connected
components. Consider any arbitrary component $V_i$.
Let $X\subset V_i\cap \Omega_1$. Now $X\subset g^{-1} (g(A))$. So,
\begin{multline*}
\mu_1(X) = \sum_{a\in X} \mu_1(a) 
= \sum_{a\in X} \sum_{b\in g(X)} \mu(a,b) \\
< \sum_{a\in g^{-1}(g(X))} \sum_{b\in g(X)} \mu(a,b) 
= \sum_{b\in g(X)} \sum_{a\in g^{-1}(g(X))} \mu(a,b) 
= \mu_2(g(X))
\end{multline*}
Using a similar argument, we can show that for $Y\subset V_i\cap
\Omega_2$, $\mu_2(Y)< \mu_1(g^{-1}(Y))$. So there is an
$\varepsilon'>0$ so that $\mu_1(X) (1+\varepsilon) \le \mu_2 (f(X)) (
1-\varepsilon )$ and $\mu_2(Y) (1+\varepsilon) \le \mu_1(f^{-1}(Y))
(1-\varepsilon)$. Let $\varepsilon$ be minimum of all such
$\varepsilon'$ where minimum is taken over all $i,X$ and $Y$.
\end{proof}

Now we quote a variant of the Marriage Lemma.
\begin{lemma}[Marriage Lemma]
\label{lemma:marraige}
For any given society $S$ between $(\Omega_1,\mu_1)$ and
$(\Omega_2,\mu_2)$, any minimal subsociety $R$ has the property that
$|\Omega_2| > |\{w\in\Omega_2:(\exists w_1,w_2\in\Omega_1)
(w_1\not=w_2\wedge w_1Rw \wedge w_2Rw)\}|$.
\end{lemma}
The proof is exactly analogous to \cite{KeaneSmorodinsky79}, see
Chapter 6 of \cite{Petersen89}.

During the construction of the isomorphism, we compute various minimal
subsocieties. There can be many such minimal subsocieties and
``inconsistent'' choices in different stages may break the
construction.  In the following subsections, we describe a way of
choosing the minimal subsocieties such that the construction goes
through.



\subsection{Construction of the isomorphism}{\label{sub-sec:assignment}}
We now have a skeleton $S$ common to two sequences $x \in
\Sigma_A^\infty$ and $y \in \Sigma_C^\infty$, and have defined an
equivalence relation on the fillers for $S$ in $\Sigma_A^*$ and
$\Sigma_C^*$ for a desired level of error. We now inductively build
societies between equivalence classes of fillers of \A and of \C and
use the marriage lemma from the preceding section to define an
isomorphism between \A and \C. A minor technical issue
arises here owing to the fact that we only have approximations of
probabilities of \A and \C during computation of canonical minimal
sub-society.

Given a skeleton $S$ of rank $r$, $r \ge 1$, and length $\ell$, let
$\F(S)$ and $\G(S)$ denote the set of its fillers in \A and \C. Given
$n\ge r$, let $\tilde\F(S,n)$ and $\tilde\G(S,n)$ denote the set of
equivalence classes with respect to the equivalence relation $\sim_n$
( i.e, $\tilde\F(S,n)=\{\tilde F_n:F\in\F(S)\}$ and
$\tilde\G(S,n)=\{\tilde G_n:G\in\G(S)\}$ ) . We denote the
$\varepsilon_n$-robust societies between $\tilde\F(S, n)$ and
$\tilde\G(s,n)$ by induction on $r$: $R_{S, n} : \tilde\F(S, n)
\leadsto \tilde\G(S, n)$ if r is odd , and $R_{S, n} : \tilde\G(S, n)
\leadsto \tilde\F(S, n)$ otherwise. The measure for every $F \in
\tilde\F(S,n)$ is $P_A(F,n)$ and that for every $G\in\tilde\G(S,n)$ is
$P_C(G, n)$.

Fix an $n$. For $r=1$, build a trivial society where each of
$\tilde\F(S,n)$ knows each of $\tilde\G(S,n)$.  Construct a minimal
$\varepsilon_n$-robust sub-society $R_{S,n}$ of the trivial society.

Now we describe the inductive construction: let $r>1$ be even. Let $S$
be a skeleton of rank $r$.  Let $S=S_1\times S_2 \times \dots\times
S_t$ be a rank $r-1$ skeleton decomposition of $S$.  Assume that we
have a procedure to define societies for all $S_i$ ranks at most $r-1$
and to any desired precision. Let us consider $R_{S,n\log
  3t}:\tilde\F(S_i,n\log 3t)\leadsto \tilde\G(S_i,n\log 3t)$ for
$i=1,2,\dots t$. Note that we are using induction only on $r$ and not
$n$ -- for a higher precision, we repeat the induction procedure from
scratch.  Consider their duals $R^*_{S,n\log 3t}:\tilde\G(S_i,n\log
3t)\leadsto\tilde\F(S_i,n\log 3t)$.  Construct the join of societies
$R:\tilde\G(S_1,n\log 3t)\times\tilde\G(S_2,n\log 3t)\times\dots\times
\tilde\G(S_t,n\log 3t)\overset{prod}\longrightarrow\tilde\F(S_1,n\log
3t)\times\tilde\F(S_2,n\log 3t)\times\dots\times \tilde\F(S_t,n\log
3t)$.

Let $\bar\F(S,n)=\tilde\F(S_1,n\log
3t)\times\dots\times\tilde\F(S_t,n\log 3t)$ and
$\bar\G(S,n)=\tilde\G(S_1,n\log
3t)\times\dots\times\tilde\G(S_t,n\cdot \log 3t)$. So,
$R:\bar\G(S,n)\leadsto \bar\F(S,n)$.

\begin{lemma}
  \label{lemma:prod-of-society}
  The $R$ constructed above is $\varepsilon_n$-robust with respect to
  measure $P_C(\cdot,n)$ and $P_A(\cdot,n)$.
\end{lemma}
We omit the proof -- it is routine to verify the conditions of robust
society hold when we approximate $P_A()$ and $P_C()$ with $P'_A()$ and
$P'_C()$ and use equation \ref{eqn:probmult}.

Since $\bar\F(S,n)$ is determined by $J_0(F,n)$ and $\tilde\F(S,n)$ is
determined by $J(F,n)$, the latter is the finer equivalence class. So
we may consider $R:\bar\G(S,n)\leadsto \tilde\F(S,n)$, where each
$\bar\F(S,n)$ is split into multiple $\tilde\F(S,n)$s and the
knowledge relation is extended accordingly. Construct \emph{the}
minimal $\varepsilon_n$-robust sub-society $U$ of $R$. From $U$,
construct $R_{S,n}:\tilde\G(S,n)\leadsto\tilde\F(S,n)$ such that
$R_{S,n}(\tilde\G(S,n))= U(\bar\G(S,n))$ where $\bar\G(S,n)$ is
uniquely determined by the finer equivalence class $\tilde\G(S,n)$.

We construct the canonical $\ve_n$-robust minimal sub-society by progressively
constructing $\ve_i$-robust minimal sub-societies for $1\le i\le n$. For 
$\ve_{i+1}$-robust minimal sub-society, we start with the  $\ve_{i}$-robust
minimal sub-society and keep removing edges from it as long as it remains
an $\ve_{i+1}$-robust society. This can be done in a computable manner,
since checking whether a finite bipartite graph is $\ve_n$-robust is
computable.


For odd $r$, we switch the role of $F$ and $G$.

Now let us describe the construction of the isomorphism: For an $x\in
K'_{r'}\cap K''_{r'}$, let $S_{x,r,i_r}$ denote the skeleton of rank
$r$ which occurs in $x$, where $i_r$ is the current central
co-ordinate.  Given an $n\in\N$, let $\tilde\F_r(x,n)$ denote the
equivalence class of fillers that occur in $x$ corresponding to
$S_r(x)$ with respect to the equivalence relation $\sim_n$. We use a
similar notation for \C where $\mathcal{G}$ replaces $\mathcal{F}$.

For $x\in\A$ such that $x\in K'_{r'}\cap K''_{r'}$ we find a large
enough even $r$ (this $r$ is computable from $r'$) such that
$\forall\bar{G}_r(x,r) \in R_{S_r(x),r}^{-1}( \tilde\F_r(x,r))$, we
have that $\bar\G_r(x,r)[{i_r}]$ is defined (it stabilizes
thenceforth).  Now the shift preserving map $\phi$
is so defined that
\begin{equation*}
  (\phi(x))[0] =
 \begin{cases}
   0 & \text{if the block of 0 containing $i_r$ is longer than 1}\\
   \bar\G_r(x,r)[i_r] & \text{otherwise}.
 \end{cases}
\end{equation*}
Note that this definition specifies all the co-ordinates of $\phi(x)$
because the we want it to be shift preserving. 

It is known that if a measure-preserving shift applied to a
Martin-L\"of random $x$ yields a Martin-L\"of random point
\cite{Nan08}, \cite{GHR09b},\cite{Shen82}. The following lemma is a
straightforward extension to Schnorr randoms. 

\begin{lemma}
Suppose $T: X \to X$ is a computable measure-preserving transformation
on a computable probability space $(X, \mathcal{F},P)$. Then the image
of every Schnorr random in $X$ under $T$ is Schnorr random.
\end{lemma}

This follows from the fact that if $\langle U_n \rangle_{n \in \N}$ is
a Schnorr layering of $X$, then so is $\langle T^{-1} U_n \rangle_{n
  \in \N}$, since $T$ is measure-preserving and computable.


This concludes the description of the algorithm for constructing
$\phi$. 


\subsection[Proof that Phi is a layerwise lower semicomputable
  isomorphism]{Proof that $\phi$ is a layerwise lower semicomputable
  isomorphism}{\label{sub-sec:proof}} 
Now we show that $\phi$ is an isomorphism and well-defined n every
Schnorr random element in \A, that $\phi^{-1}$ is well-defined for
every Schnorr random element in \C, and that the candidate isomorphism
$\phi$ is Schnorr layerwise lower semicomputable.

\subsubsection[Phi is isomorphic]{$\phi$ is isomorphic}
We show that we can always find an $r$ sufficiently large to stabilize
the construction of the society, Let us consider the case when for a
given $r$-rank skeleton $S$, $n$ is so large that $R_{S,n}$ stabilizes
(\emph{i.e.}, it remains unchanged for any larger $n$) -- such a $n$
exists due to Lemma \ref{lemma:robust-suffice} and the fact that
$J(\cdot,n)$s are non-decreasing in $n$ and bounded above. Call such
stabilized society $R_S:\bar\G(S)\leadsto\tilde\F(S)$. Then the
following result holds.
\begin{lemma}[Assignment Lemma]
\label{lemma:assignment}
If $x\in\A$ such that $x\in G_{r'}\cap G'_{r'}$ with $x[0]$ not 
contained in a block of $0$ longer than $m$, then there is an even
$r$, computable from $r'$, such that
\begin{enumerate}
\item With respect to the society
  $R_{S_r(x)}:\bar\G(S_r(x))\leadsto \tilde\F(S_r(x))$,
  $R_{S_r(x)}^{-1}(\tilde\F_r(x))$ is a singleton, say,
  $\bar\G_r(x)$.
\item $i_r(x)\in J_0(\bar\G_r(x))$.
\end{enumerate}
\end{lemma}
We omit the proof of this lemma -- it is similar to the proof of the
assignment lemma given in \cite{Petersen89} where we use the estimates
given by lemma \ref{lemma:filler}.

Now we show how the above lemma ensures the existence of the map
$\phi$ for every $x\in G_{r'}\cap G'_{r'}$. If the
co-ordinate $i_r$ is part of a block of $0$ of length at least 2,
then we are done. Otherwise, the above lemma shows that for each $x
\in G_{r'}\cap G'_{r'}$, there is a sufficiently large $r$, computable
from $r'$ such that for all sufficiently large
$n$, $(\bar\G_r(x))_{i_r}$ becomes fixed -- this is defined to be
$\phi(x)[0]$. Let $r_1$ be greater than $n$ and $r$. Since
$R_{S_{r_1}(x),r_1}$ is derived from $R_{S_r(x),n}$ (via the
construction of consistent minimal sub-society), we have that all
$\bar\G_{r_1}(x,r_1)$ which know some $\tilde\F_{r_1}(x,r_1)$ 
have the coordinate $i_{r}$ fixed with same symbol
$(\bar\G_r(x))_{i_r}$. Hence at this $r_1$ we can level off the
inductive construction we can compute $\phi(x)[0]$.

Finally we show that $\phi$ is indeed an isomorphism. The map is by
construction measurable, and shift-invariant. We only need to show
that it is measure-preserving. We use a similar technique as in the
original proof. Consider $x\in\C$ specified by fixing $z$ consecutive
co-ordinates for some $z$. We show that for all $Y \in \Sigma_C^{z}$,
$P_A(\phi^{-1}(Y))\ge P_C(Y)$. Consequently, $\phi$ is
measure-preserving on the algebra $\Sigma_C^z$. This is sufficient,
since elements of $\Sigma_C^z$ over all $z$ generate the
$\sigma$-algebra $\Sigma_C^\infty$.

Let $X=\{ x \in \Sigma_A^\infty \mid x[k \dots
  z+k]=c_{i_k}c_{i_{k+1}}\dots c_{i_{z+k}}\}$. Since both \A and \C
are stationary and $\phi$ is shift preserving, we can assume that
$k=0$. Consider $x\in T_C^{-(k+1)}X$, i.e., the symbols in positions
$-k-1$ to $-1$ match those in corresponding places of $C$. Now,
consider a cylinder $xa$ for $a\in\Sigma_C$.  Clearly, $X=\cup_{a\in
  C}Xa$. Now we use the assignment lemma on cylinder $Xa$ to argue
about the measures.  The assignment 
theorem \cite{Royd88} implies for all $x \in G_{r'}\cap G'_{r'}$ there
is an $r_1$ such that $(G_{r'}\cap G'_{r'})^c$ has measure
$\delta_{r'}$ and we can find the assignment for $\phi(x)[0]$ in the
$r_1$ level off the inductive construction.  Note, $\delta_{r'}\to 0$
as $r'\to\infty$. So,
\begin{multline*}
P_C(x)=\sum_{a\in \Sigma_C}P_C(Xa) \le \sum_{\bar\F_{r_1}(Xa)\in
  R_{S_{r_1}(Xa),r_1}^{-1}(\tilde\G_{r_1}(Xa))} P_A(\bar\F_r(Xa),r_1) 
+ \delta_{r'}\\ 
\le P_A(\phi^{-1}(X)) (1-\varepsilon_{r_1})+\delta_{r'},
\end{multline*}
since the map $\phi$ respects society and we consider
$\varepsilon_{r_1}$ robust societies in level $r_1$.  Since $x \in
G_{s}\cap G'_{s}$ for all $s\ge r'$, we have $P_C(X)\le
P_A(\phi^{-1}(X))$.

\subsubsection[Schnorr Layerwise Lower Semicomputability of Phi]{Schnorr Layerwise Lower Semicomputability of $\phi$}
In this section, we recapitulate the major steps in the construction
of the isomorphic map $\phi$ and show that it is Schnorr layerwise
lower semicomputable. This yields, as a corollary, that it is defined
for every Schnorr random sequence $x \in \A$. We conclude by proving
that $\phi(x) \in \C$ is a Schnorr random as well.

We show that there is a Schnorr layering $\langle K^A_r
\rangle_{r=1}^{\infty}$ of $\A$ such that the following holds. For
every $x \in K^A_r$, there is a central cylinder $x[\;-m_r+1 \dots 0
  \dots m_r-1\;]$ mapped to a central cylinder $y[\;-m_r+1 \dots 0
  \dots m_r+1\;]$ such that $P_A(x[-m_r+1 \dots 0 \dots m_r-1])$ is
approximately $P_C(y[-m_r+1 \dots 0 \dots m_r-1])$.

To see this, note that the Schnorr layering $\langle K'_r \cap K''_r
\rangle_{r=1}^{\infty}$ of \A, where $\langle K'_r
\rangle_{r=1}^{\infty}$ is the Schnorr layering of \A in the Skeleton
Lemma and $\langle K''_r \rangle_{r=1}^{\infty}$ is its Schnorr
layering in the Filler Lemma, has the following property. For every $r
\in \N$ and $x \in G_r \cap G_r'$, there is a central skeleton of $x$
of rank $r$ and length $L_r$, for which every filler $F \in
\Sigma_A^{L_r}$ obeys the probability bounds in the filler lemma.

Similarly, there is a Schnorr layering of \C which has the following
property. For every $r \in \N$ and $y$ in the $r^{\text{th}}$ layer,
there is a central skeleton of $x$ of rank $r$ and length $L_r$, for
which every filler $G \in \Sigma_C^{L_r}$ obeys the probability bounds
in the filler lemma.

Then we create a bipartite graph among the equivalence classes
$\tilde{F}_n$ and $\tilde{G}_n$ of fillers in $\Sigma^{L_r}$ and
$\Sigma_C^{L_r}$, and build the canonical $\varepsilon_n$-robust
minimal subsociety. This is a computable process, since the societies
are finite. The assignment lemma yields us a layerwise lower
semicomputation of the central co-ordinate $\phi(x)[0]$.

Let $T_A$ and $T_C$ be the shifts associated with \A and \C,
respectively. If $x$ is Schnorr random in \A, the computabilty and
measure-preservation of $T_A$ ensure that $T_A^ix$, $i \in \Z$ is also
Schnorr random in \A. Hence for all large enough ranks $r'$, $T_A^ix
\in K^{A}_{r'}$. Noting that $x[i] = (T_A^ix)[0]$ and that $\phi$ is a
factor map, we see that
$$(\phi \circ T_A^i (x))[0] = (T_C^i \circ \phi(x))[0] =
(\phi(x))[i],$$ we see that all co-ordinates $\phi(x)[-m+1 \dots 0
  \dots m+1]$ will be fixed for all large enough ranks $K^A_r$. This
is an iteration over a Schnorr layerwise lower semicomputable
function, hence is Schnorr layerwise lower semicomputable.  For
$\phi^{-1}$, the same argument can be carried out on the dual graph.

Hence the maps $\phi$ and $\phi^{-1}$ thus constructed are Schnorr
layerwise lower semicomputable and can be computed for all Schnorr
random points.

\begin{lemma}
Let $t_A : \Sigma_A^\infty \to [0,\infty]$ be a Schnorr layerwise
$P_A$-integrable test. Then $t_C' = t_A \circ \phi^{-1}$ is a Schnorr
layerwise $P_C$-integrable test. Conversely, if $t_C : \Sigma_A^\infty
\to [0,\infty]$ be a Schnorr layerwise $P_C$-integrable test. Then
$t_A' = \phi~\circ~t_A$ is a Schnorr layerwise $P_A$-integrable test.
\end{lemma}
\begin{proof}
The function $t_C' = t_A\phi^{-1}$ is layerwise
lowersemicomputable. Also, $\int t_C' dP_C = \int t_A \circ \phi^{-1}
dP_A$, since $\phi$ is a measure-preserving isomorphism. Hence $\int
t_C' dP_C$ is finite. If $s_1, s_2, \dots$ is the computable sequence
of step functions witnessing the $L^1$ computability of $t_A$, then
$s_1 \circ \phi^{-1}, s_2 \circ \phi^{-1}, \dots$ witnesses the $L^1$
computability of $t_C'$.  Thus $t_C'$ is a Schnorr layerwise
$P_C$-integrable test.

The proof in the converse direction is similar.
\end{proof}

\begin{corollary}
$x \in $\A is Schnorr random if and only if $\phi(x) \in$
  \C is Schnorr random, and $y \in $\C is Schnorr random if
  and only if $\phi^{-1}(y) \in$ \A is Schnorr random.
\end{corollary}
\begin{proof}
Let $t_A, t_A', t_C$ and $t_C'$ be as in the previous lemma. If
$t_A(\phi^{-1}(y))=\infty$, then $t_C'(y) = \infty$ implying that $y$
is not Schnorr random in \C.

Conversely, by a similar argument, we see that for $x \in \A$ such
that $\phi(x) \in \C$ is defined, if $\phi(x)$ is not Schnorr
random in \C, then $x$ is not Schnorr random in \A.
\end{proof}



\section{Computable isomorphisms}
\label{secn:computable_transformation}
Recall that a \emph{homeomorphism} is a continuous bijection whose
inverse is also continuous. It is known (see \cite{Petersen89},
section 6.5, page 301, excercise 2):
\begin{lemma}
\label{lemma:computable-isomorphism}
Suppose $\X$ and $\Y$ be two Bernoulli systems with the same
entropy. If $\phi: X \to Y$ is a measure-preserving homeomorphism,
then $\mu$ and $\nu$ are permutations of each other.
\end{lemma}
Since total computable functions are continuous, it follows that only
trivial computable isomorphisms exist between two computable
dynamical systems. This partly justifies layerwise lower
semicomputability as a notion of appropriate power for constructing
the isomorphism between the systems.

\section{Comparison of the results}
Ornstein showed that a process satisfying a weaker condition,
\emph{viz.} a \emph{finitely determined system} with entropy $H$ is
isomorphic to some Bernoulli process with entropy $H$. Thus elements
of a much broader class of processes are isomorphic to Bernoulli
systems of equal entropy, the latter being intuitively the \emph{most
  random systems} possible. Several ``deterministic'' dynamical
systems have been shown to be finitely determined (for a survey, see
Ornstein\cite{Ornstein13}), leading to the interpretation that all
such systems are, intuitively, encodings of the most random possible
systems. However, to demonstrate this, we need isomorphic maps which
are termed \emph{stationary codes}. \cite{Shields98} Rudolph has
proved a characterization of systems \emph{finitarily isomorphic} to
each other \cite{Rudolph81}, showing that if we restrict our codes to
\emph{finitary codes}, there are weakly Bernoulli systems and finitely
determined systems which cannot be isomorphic to any Bernoulli system
with the same entropy.

We show that computable mixing Markov systems of equal entropy have a
\emph{layerwise lower semicomputable} isomorphism. Thus the targets of
our isomorphisms are not intuitively as random as that of the Ornstein
construction. However, our code has a stronger computability property
than Ornstein's original construction and the maps in Rudolph's
characterization of finitary isomorphism.

Rudolph's characterization of systems finitarily isomorphic to
Bernoulli systems uses the notion of \emph{conditional block
  independence}. We leave open whether there is a similar
characterization of computable systems which are layerwise isomorphic
to a computable mixing Markov system.
\section*{Acknowledgments}
The authors thank Mathieu Hoyrup and Jason Rute for valuable
discussions and anonymous referees for their suggestions.



\bibliography{main,dim,dimrelated,random,fair001,biblio}
\end{document}